\documentclass[journal, onecolumn]{IEEEtran}

\usepackage{amsmath,amsfonts,amsthm,amssymb}
\usepackage{setspace}
\usepackage{fancyhdr}
\usepackage{extramarks}
\usepackage{chngpage}
\usepackage[usenames,dvipsnames]{color}
\usepackage{ifthen}
\usepackage{courier}
\usepackage{graphicx}
\usepackage{textcomp}
\usepackage{amssymb}
\usepackage{cancel}
\usepackage{xcolor}
\usepackage[bookmarks=false,hidelinks,hyperfootnotes=false]{hyperref}
\usepackage[numbers,sort&compress]{natbib}
\usepackage{mathtools}
\usepackage{enumitem}

\newcommand{\DoF}{\mathsf{DoF}}

\newtheorem{defi}{\textbf{Definition}}
\newtheorem{theo}{\textbf{Theorem}}

\newtheorem{lem}{\textbf{Lemma}}

\theoremstyle{definition}
\newtheorem{ex}{Example}
\newtheorem*{illust}{Illustration}
\newtheorem{remk}{Remark}

\usepackage{algorithm}
\usepackage[noend]{algpseudocode}\usepackage{subfigure}

\makeatletter
\def\BState{\State\hskip-\ALG@thistlm}
\newcommand{\subalign}[1]{%
  \vcenter{%
    \Let@ \restore@math@cr \default@tag
    \baselineskip\fontdimen10 \scriptfont\tw@
    \advance\baselineskip\fontdimen12 \scriptfont\tw@
    \lineskip\thr@@\fontdimen8 \scriptfont\thr@@
    \lineskiplimit\lineskip
    \ialign{\hfil$\m@th\scriptstyle##$&$\m@th\scriptstyle{}##$\crcr
      #1\crcr
    }%
  }
}
\makeatother

\newcommand\numberthis{\addtocounter{equation}{1}\tag{\theequation}}
\newcommand*{\Scale}[2][4]{\scalebox{#1}{$#2$}}

\allowdisplaybreaks

\usepackage{scalerel,stackengine}
\stackMath
\newcommand\reallywidehat[1]{%
\savestack{\tmpbox}{\stretchto{%
  \scaleto{%
    \scalerel*[\widthof{\ensuremath{#1}}]{\kern-.6pt\bigwedge\kern-.6pt}%
    {\rule[-\textheight/2]{1ex}{\textheight}}
  }{\textheight}%
}{0.5ex}}%
\stackon[1pt]{#1}{\tmpbox}%
}

\begin{document}
\include{commands}
\title{Fundamental Limits of Cache-Aided \\ Interference Management\footnote{A short version of this paper will be presented at the IEEE International Symposium on Information Theory (ISIT), 2016. This work is the outcome of a collaboration that started while N. Naderializadeh was a research intern at Bell Labs. This work is in part supported by NSF grants CAREER 1408639, NETS-1419632, EARS-1411244, and ONR award N000141612189.}}

\author{Navid Naderializadeh$^{*}$, Mohammad Ali Maddah-Ali$^{\dagger}$, and A. Salman Avestimehr$^{*}$\\
$^{*}$Department of Electrical Engineering, University of Southern California, Los Angeles, CA, USA \\ 
$^{\dagger}$Nokia Bell Labs, Holmdel, NJ, USA\\
E-mails: \href{mailto:naderial@usc.edu}{naderial@usc.edu}, \href{mailto:mohammad.maddah-ali@nokia.com}{mohammad.maddah-ali@nokia.com}, \href{mailto:avestimehr@ee.usc.edu}{avestimehr@ee.usc.edu}
}

\maketitle
\begin{abstract}

We consider a system, comprising a library of $N$ files (e.g., movies) and a wireless network with $K_T$ transmitters, each equipped with a local cache of size of $M_T$ files, and $K_R$ receivers, each equipped with a local cache of size of $M_R$ files. Each receiver will ask for one of the $N$ files in the library, which needs to be delivered. The objective is to design the cache placement (without prior knowledge of receivers' future requests) and the communication scheme to maximize the throughput of the delivery. In this setting, we show that the sum degrees-of-freedom (sum-DoF) of $\min\left\{\frac{K_T M_T+K_R M_R}{N},K_R\right\}$ is achievable, and this is within a factor of 2 of the optimum, under one-shot linear schemes. This result shows that (i) the one-shot sum-DoF scales \emph{linearly} with the \emph{aggregate cache size in the network} (i.e., the cumulative memory available at \emph{all nodes}), (ii) the transmitters' caches and receivers' caches contribute equally in the one-shot sum-DoF, and (iii) caching can offer a throughput gain that scales linearly with the size of the network. 

To prove the result, we propose an achievable scheme that exploits the redundancy of the content at transmitters' caches to  cooperatively zero-force some outgoing interference, and  availability of the unintended content at the receivers' caches to cancel (subtract) some of the incoming interference. We develop a particular pattern for cache placement that maximizes the overall gains of cache-aided transmit and receive interference cancellations. For the converse, we present an integer optimization problem which minimizes the number of communication blocks needed to deliver any set of requested files to the receivers. We then provide a lower bound on the value of this optimization problem, hence leading to an upper bound on the linear one-shot sum-DoF of the network, which is within a factor of 2 of the achievable sum-DoF.
\end{abstract}

\section{Introduction}\label{sec:intro}


Over the last decade, video delivery has emerged as the main driving factor of the wireless traffic. In this context, there is often a large library of pre-recorded content (e.g. movies), out of which, users may request to receive a specific file. One way to reduce the burden of this traffic is to employ memories distributed across the networks and closer to the end users to prefetch some of the popular content. This can help system to deliver the content with higher throughput and less delay.
 
%


As a result, there have been significant interests in both academia and industry in
characterizing the impact of caching on the performance of communication networks (see, e.g. \cite{golrezaei12, maddahali_caching, multi_server_caching, liu15, sengupta2015cache, ji2013wireless, MaddahAli_3userIC_Caching,azari2015hypergraph,park2016joint,
afshang2015fundamentals,tao2015content,ugur2015cloud,peng2015backhaul}). In particular, in a network with only one transmitter broadcasting to several receivers, it was shown in \cite{maddahali_caching} that local delivery attains only a small fraction of the gain that caching can offer, and by designing a particular pattern in cache placement at the users and exploiting coding in delivery, a significantly larger \emph{global} throughput gain can be achieved, which is a function of the entire cache throughout the network. This also demonstrates that the gain of caching scales with the size of the network. As a follow-up, this work has been extended to the case of multiple transmitters in~\cite{multi_server_caching}, where it was shown that the gain of caching can be improved if several transmitters have access to the entire library of files. Caching at the transmitters was also considered in~\cite{liu15,sengupta2015cache} and used to induce collaboration between transmitters in the network.  It is also shown in~\cite{MaddahAli_3userIC_Caching} that caches at the transmitters can improve load balancing and increase the opportunities for interference alignment. More recently, the authors in \cite{azari2015hypergraph} evaluated the performance of cellular networks with edge caching via a hypergraph coloring problem. Furthermore, in \cite{park2016joint}, the authors studied the problem of maximizing the delivery rate of a fog radio access network for arbitrary prefetching strategies.

In this paper, we consider a general network setting with caches at both transmitters and receivers, and demonstrate how one can utilize caches at both transmitters and receivers to manage the interference and enhance the system performance in the physical layer. In particular, we consider a library of $N$ files and a wireless network with $K_T$ transmitters and $K_R$ receivers, in which each transmitter and each receiver is equipped with a cache memory of a certain size. In particular, each transmitter and each receiver can cache up to $M_T$ and $M_R$ files, respectively. The system operates in two phases. The first phase is called the prefetching phase, where each cache is populated up to its limited size from the content of the library. This phase is followed by a delivery phase, where each user reveals its request for a file in the library. The transmitters then need to deliver the requested files to the receivers. Note that in the prefetching phase, the system is still unaware of the files that the receivers will request in the delivery phase. The goal is to design the cache contents in the prefetching phase and communication scheme in the delivery phase to achieve the maximum throughput for arbitrary set of requested files. Due to their practical appeal, in this work we focus on one-shot linear delivery strategies. Interestingly, many of the previous works on caching have relied on one-shot schemes for content delivery (see, e.g. \cite{liu15,azari2015hypergraph}).

Our main result in this paper is the characterization of the one-shot linear sum degrees-of-freedom (sum-DoF) of the network, i.e., number of the receivers that can be served interference-free simultaneously, within a factor of 2 for all system parameters. In fact, we show that the one-shot linear sum-DoF of $\min\left\{\frac{K_T M_T+K_R M_R}{N},K_R\right\}$ is achievable, and this is within a factor of 2 of the optimum. This result shows that the one-shot linear sum-DoF of the network grows linearly with the aggregate cache size in the network (i.e., the cumulative memory available at all nodes). It also implies that caches at the transmitters' side are equally valuable as the caches on the receivers' side in the one-shot linear sum-DoF of the network. Our result, therefore, establishes a fundamental limit on the performance of one-shot delivery schemes for cache-aided interference management.

To achieve the aforementioned sum-DoF, we propose a particular pattern in cache placement so that each piece of each file in the library is available in the caches of $\frac{K_T M_T}{N}$ transmitters and $\frac{K_R M_R}{N}$ receivers. Once caching is done this way, we can show that for delivering any set of requested contents to the receivers, $\min\left\{\frac{K_T M_T+K_R M_R}{N},K_R\right\}$ of the receivers can be served at each time, interference-free. This gain is achieved by simultaneously exploiting the opportunity of collaborative interference cancellation (i.e. zero-forcing) at the transmitters' side and opportunity of eliminating known interference contributions at the receivers' side. The first opportunity is created by caching the pieces of each file at several transmitters. The second opportunity is available since pieces of a file requested by one user has been cached at some other receivers, and thus do not impose interference at those receivers effectively. Our proposed cache placement pattern maximizes the overall gain achieved by these opportunities for any arbitrary set of receiver requests and this gain can be achieved even with a simple one-shot linear delivery scheme.

%

Moreover, we demonstrate that our achievable sum-DoF is within a factor of 2 of the optimal sum-DoF for one-shot linear schemes. To prove the outer bound, we take a four-step approach in order to lower bound the number of communication blocks needed to deliver any set of requested files to the receivers. First, we show that the network can be converted to a virtual MISO interference channel in each block of communication. Using this conversion, we next write an integer optimization problem for the minimum number of communication blocks needed to deliver a fixed set of requests for a given caching realization. We then show how we can focus on average demands instead of the worst-case demands to derive an outer optimization problem on the number of communication blocks optimized over the caching realizations. Finally, we present a lower bound on the value of the aforementioned optimization problem, which leads to the desired upper bound on the one-shot linear sum-DoF of the network. This result illustrates that in this setting, caches at transmitters' side are equally valuable as caches at receivers' side. It also shows that caching offers a throughput gain that scales linearly with the size of the network.

The rest of the paper is organized as follows. We present the problem formulation in Section \ref{sec:model}. We state the main result in Section \ref{sec:result}. We prove the achievability of our main result in Section \ref{sec:ach} and the converse in Section \ref{sec:converse}. Finally, we conclude the paper in Section \ref{sec:conc}.



\section{Problem Formulation}\label{sec:model}
In this section, we first provide a high-level description of the problem setting and the main parameters in the system model, and then we present a detailed description of the problem formulation.



\subsection{Problem Overview}

Consider a wireless network, as illustrated in Figure \ref{fig:network_model}, with $K_T$ transmitters and $K_R$ receivers, and also a library of $N$ files, each of which contains $F$ packets, where each packet is a vector of $B$ bits. Each node in the network is equipped with a local cache memory of a certain size that can be used to cache contents arbitrarily from the library before the receivers reveal their requests and communication begins. In particular, each transmitter and each receiver is equipped with a cache of size $M_T F$ and $M_R F$ packets, respectively.

\begin{figure}[hbt]
\centering
\includegraphics[trim=0in 0in 0in 0in, clip, width=0.5\textwidth]{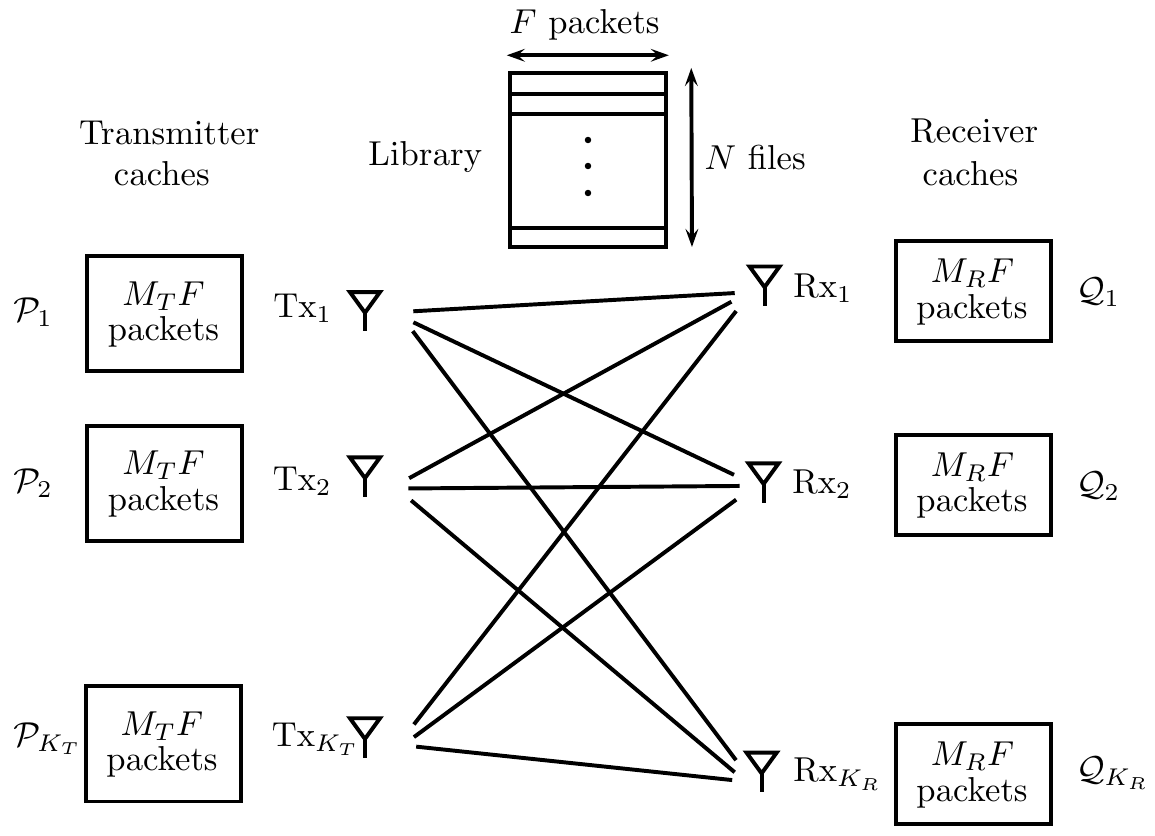}
\caption{Wireless network with $K_T$ transmitters and $K_R$ receivers, where each transmitter and each receiver caches up to $M_T F$ packets and $M_R F$ packets, respectively, from a library of $N$ files, each composed of $F$ packets.}
\label{fig:network_model}
\end{figure}

We assume that the system operates in two phases, namely the \emph{prefetching phase} and the \emph{delivery phase}. In the prefetching phase, each node can cache contents arbitrarily from the library subject to its cache size constraint. In particular, each transmitter selects up to $M_T F$ packets out of the entire library to store in its cache, and each receiver selects up to $M_R F$ packets out of the entire library to store in its cache. In the delivery phase, each receiver requests an arbitrary file from the library. Since each receiver may have cached parts of its desired file in the prefetching phase, the transmitters need to deliver the rest of the requested packets to the receivers over the wireless channel.

We assume that at each time, the transmitters employ a one-shot linear scheme, where a subset of requested packets are selected to be delivered interference-free to a corresponding subset of receivers. Each transmitter transmits a linear combination of the subset of the selected packets which it has cached in the prefetching phase. The interference is cancelled with the aid of cached contents as follows.  Since each requested packet may be cached at multiple transmitters, the transmitters can collaborate in order to zero-force the outgoing interference of that packet at some of the unintended receivers. Moreover, the receivers can also use their cached packets as side information to eliminate the remaining incoming interference from to undesired packets. 
Our objective is to design a cache placement scheme and a delivery scheme which maximize the number of packets that can be delivered at each time interference-free.


%

In this setting, we define the one-shot linear sum-degrees of freedom as the ratio of the number of delivered packets over the number of blocks needed for communicating those packets for any set of receiver demands. Finally, we define the one-shot linear sum-DoF of the network, denoted by $\DoF_{\text{L,sum}}^*(N,M_T,M_R)$, as the maximum achievable one-shot linear sum-DoF over all caching realizations.

\subsection{Detailed Problem Description}


 We consider a discrete-time additive white Gaussian noise channel, as illustrated in Figure \ref{fig:network_model}, with $K_T$ transmitters denoted by $\{\text{Tx}_i\}_{i=1}^{K_T}$ and $K_R$ receivers denoted by $\{\text{Rx}_i\}_{i=1}^{K_R}$. The communication at time $t$ over this channel is modeled by 
\begin{align}
Y_j(t)=\sum_{i=1}^{K_T} h_{ji} X_i(t)+Z_j(t),
\end{align}
where $X_i(t)\in\mathbb{C}$ denotes the signal transmitted by $\text{Tx}_i, i\in[K_T]\triangleq\{1,...,K_T\}$ and $Y_j(t)$ denotes the receive signal by $\text{Rx}_j, j\in[K_R]$. Moreover, $h_{ji}\in\mathbb{C}$ denotes the channel gain from $\text{Tx}_i$ to $\text{Rx}_j$, assumed to stay fixed over the course of communication, and $Z_j(t)$ denotes the additive white Gaussian noise at $\text{Rx}_j$ at time slot $t$, distributed as $\mathcal{CN}(0,1)$.
The transmit signal at $\text{Tx}_i, i\in[K_T]\triangleq\{1,...,K_T\}$
is subject to  the power constraint $\mathbb{E}\left[|X_i(t)|^2\right]\leq P$. 

We assume that each receiver will request an arbitrary file out of a library of $N$ files $\{W_n\}_{n=1}^{N}$, which should be delivered by the transmitters. Each file $W_n$ in the library contains $F$ packets $\{\mathbf{w}_{n,f}\}_{f=1}^{F}$, where each packet is a vector of $B$ bits; i.e., $\mathbf{w}_{n,f}\in\mathbb{F}_2^B$. Furthermore, we assume that each node in the network is equipped with a cache memory of a certain size that can be used to cache arbitrary contents from the library before the receivers reveal their requests and communication begins. In particular, each transmitter and each receiver is equipped with a cache of size $M_T F$ and $M_R F$ packets, respectively.

We assume that the network operates in two phases, namely the prefetching phase and the delivery phase, which are described in more detail as follows.





\emph{Prefetching Phase:} In this phase, each node can store an arbitrary subset of the packets from the files in the  library up to its cache size. In particular, each transmitter $\text{Tx}_i$ chooses a subset $\mathcal{P}_i$ of the $NF$ packets in the library, where $|\mathcal{P}_i|\leq M_T F$, to store in its cache. Likewise, each receiver $\text{Rx}_i$ stores a subset $\mathcal{Q}_i$ of the packets in the library, where $|\mathcal{Q}_i|\leq M_R F$. Caching is done at the level of whole packets and we do not allow breaking the packets into smaller subpackets. Also, this phase takes place unaware of the receivers' future requests.




\emph{Delivery Phase:} In this phase, each receiver $\text{Rx}_j, j\in[K_R]$, reveals its request for an arbitrary file $W_{d_j}$ from the library for some $d_j\in[N]$. We let $\mathbf{d}=[d_1 ~ ... ~ d_{K_R}]^T$ denote the demand vector. 
Depending on the demand vector $\mathbf{d}$ and the cache contents, each receiver has already cached some packets of its desired file and there is no need to deliver them.
The transmitters will be responsible for delivering the rest of the requested packets to the receivers. In order to make sure that any piece of content in the library is stored at the cache of at least one transmitter in the network, we assume that the transmitter cache size satisfies $K_T M_T\geq N$.

Each transmitter first employs a random Gaussian coding scheme $\psi: \mathbb{F}_2^B \rightarrow \mathbb{C}^{\tilde{B}}$ of rate $\log P+o(\log P)$ to encode each of its cached packets into a \emph{coded packet} composed of $\tilde{B}$ complex symbols, so that each coded packet carries one degree-of-freedom (DoF).
We denote the coded version of each packet $\mathbf{w}_{n,f}$ in the library by $\tilde{\mathbf{w}}_{n,f}\triangleq\psi(\mathbf{w}_{n,f})$. Afterwards, the communication takes place over $H$ blocks, each of length $\tilde{B}$ time slots. In each block $m\in[H]$, the goal is to deliver a subset of the requested packets, denoted by $\mathcal{D}_m$, to a subset of receivers, denoted by $\mathcal{R}_m$, such that each packet in $\mathcal{D}_m$ is intended to exactly one of the receivers in $\mathcal{R}_m$. In addition, the set of transmitted packets in all blocks and the cache contents of the receivers should satisfy
\begin{align}
 \{\mathbf{w}_{d_j,f}\}_{f=1}^{F} \subset \left(\bigcup_{m=1}^{H}\mathcal{D}_m \right) \cup \mathcal{Q}_j, \ \forall j \in [K_R],
\end{align}
which implies that for any receiver $\text{Rx}_j, j\in[K_R]$, each of its requested packets should be either transmitted in one of the blocks or already stored in its own cache.

In each block $m\in[H]$, we assume a one-shot linear scheme where each transmitter transmits an arbitrary linear combination of a subset of the coded packets in $\mathcal{D}_m$ that it has cached. Particularly, $\text{Tx}_i, i\in[K_T]$ transmits $\mathbf{x}_i[m]\in\mathbb{C}^{\tilde{B}}$,  where
\begin{align}
\label{eq:beamforming}
\mathbf{x}_i[m]=\sum_{\substack{(n,f):\\\mathbf{w}_{n,f}\in\mathcal{P}_i \cap \mathcal{D}_m}} v_{i,n,f}[m] ~ \tilde{\mathbf{w}}_{n,f},
\end{align}
and $v_{i,n,f}[m]$'s denote the complex beamforming coefficients that $\text{Tx}_i$ uses to linearly combine its coded packets in block $m$.  




On the receivers' side, the received signal of each receiver $\text{Rx}_j\in\mathcal{R}_m$ in block $m$, denoted by $\mathbf{y}_j[m]\in\mathbb{C}^{\tilde{B}}$, can be written as
\begin{align}
{\mathbf{y}_j[m]=\sum_{i=1}^{K_T} h_{ji} \mathbf{x}_i[m]+\mathbf{z}_j[m],}
\end{align}
where $\mathbf{z}_j[m]\in\mathbb{C}^{\tilde{B}}$ denotes the noise vector at $\text{Rx}_j$ in block $m$. Then, receiver $\text{Rx}_j$ will use the contents of its cache to cancel (subtract out) the interference of some of undesired packets in $\mathcal{D}_m$, if they exist in its cache.   In particular, each receiver $\text{Rx}_j \in \mathcal{R}_m $, forms a linear combination $\mathcal{L}_{j,m}$, as
\begin{align}
\label{eq:decoding1}
\mathcal{L}_{j,m}(\mathbf{y}_j[m], \tilde{\mathcal{Q}}_j  ) 
\end{align}
to recover $\tilde{\mathbf{w}}_{d_j,f} \in \mathcal{D}_m$, where $\tilde{\mathcal{Q}}_j$ denotes the set of coded packets cached at receiver $\text{Rx}_j$.

The communication in block $m\in H$ to transmit the packets in $\mathcal{D}_m$  is successful, if there exist linear combinations \eqref{eq:beamforming} at the transmitters' side and \eqref{eq:decoding1} at receivers' side, such that for all $\text{Rx}_j\in\mathcal{R}_m$, 
\begin{align}\label{eq:decoding}
\mathcal{L}_{j,m}(\mathbf{y}_j[m], \tilde{\mathcal{Q}}_j )  = \tilde{\mathbf{w}}_{d_j,f}+\mathbf{z}_j[m]. 
\end{align}
The channel created in \eqref{eq:decoding} is a point-to-point channel, whose capacity is $\log P + o(\log P)$. Hence, since each coded packet $\tilde{\mathbf{w}}_{d_j,f}$ is coded with rate $\log P + o(\log P)$, it can be decoded with vanishing error probability as $B$ increases. We assume that the communication continues for $H$ blocks until all the desired packets are successfully delivered to all receivers.

Since each packet carries one degree-of-freedom, the one-shot linear sum-degrees-of-freedom (sum-DoF) of $|\mathcal{D}_m|$ is achievable in each block $m\in [H]$. This implies that throughout the $H$ blocks of communication, the one-shot linear sum-DoF of $\frac{\left|\bigcup_{m=1}^H \mathcal{D}_m \right|}{H}$ is achievable.
Therefore, for a given caching realization, we define the one-shot linear sum-DoF to be maximum achievable one-shot linear sum-DoF for the worst case demands; i.e.,
\begin{align}\label{eq:d_sum_cache}
{\DoF_{\text{L,sum}}^{\left(\{\mathcal{P}_i\}_{i=1}^{K_T},\{\mathcal{Q}_i\}_{i=1}^{K_R}\right)}= \adjustlimits \inf_{\mathbf{d}} \sup_{H,\{\mathcal{D}_m\}_{m=1}^H} \frac{\left|\bigcup\limits_{m=1}^H \mathcal{D}_m \right|}{H}.}
\end{align}

This leads us to the definition of the one-shot linear sum-DoF of the network as follows.

\begin{defi}\label{def:sumdof}
For a network with a library $N$ files, each containing $F$ packets, and cache size of $M_T$ and $M_R$ files at each transmitter and receiver, respectively, we define the one-shot linear sum-DoF of the network as the maximum achievable one-shot linear sum-DoF over all caching realizations; i.e.,
\begin{align}
{\DoF_{\emph{L,sum}}^*(N,M_T,M_R)=} & {\sup_{\{\mathcal{P}_i\}_{i=1}^{K_T},\{\mathcal{Q}_i\}_{i=1}^{K_R}}  \DoF_{\emph{L,sum}}^{\left(\{\mathcal{P}_i\}_{i=1}^{K_T},\{\mathcal{Q}_i\}_{i=1}^{K_R}\right)}}\\
&\quad\quad~~\emph{s.t.} ~\ \ \Scale[.93]{|\mathcal{P}_i|\leq M_T F, ~\forall i\in[K_T]}\\
&\hspace{.68in} {|\mathcal{Q}_i|\leq M_R F, ~\forall i\in[K_R],}
\end{align}
where $\DoF_{\emph{L,sum}}^{\left(\{\mathcal{P}_i\}_{i=1}^{K_T},\{\mathcal{Q}_i\}_{i=1}^{K_R}\right)}$ is defined in \eqref{eq:d_sum_cache}.
\end{defi}

\section{Main Result and its Implications}\label{sec:result}

In this section, we present our main result on the one-shot linear sum-DoF of the network and its implications.

\begin{theo}\label{thm:main}
For a network with a library of $N$ files, each containing $F$ packets, and cache size of $M_T$ and $M_R$ files at each transmitter and each receiver, respectively, the one-shot linear sum-DoF of the network, as defined in Definition \ref{def:sumdof}, satisfies
\begin{align}\label{eq:ach_dof}
\min\left\{\frac{K_T M_T+K_R M_R}{N},K_R\right\} \leq \DoF_{\emph{L,sum}}^*(N,M_T,M_R)\leq \min\left\{2\frac{K_T M_T+K_R M_R}{N},K_R\right\},
\end{align}
for sufficiently large $F$.
\end{theo}

In the following, we highlight the implications of Theorem \ref{thm:main} and its connections to some prior works:


\begin{enumerate}

\item (\textit{Within a factor of 2 characterization}) The upper bound in (\ref{eq:ach_dof}) is within a factor of 2 of the lower bound in (\ref{eq:ach_dof}). Therefore, Theorem \ref{thm:main} characterizes the one-shot linear sum-DoF of a cache-aided wireless network to within a factor of 2, for all system parameters.

\item (\textit{Aggregate cache size matters}) The one-shot linear sum-DoF characterized in Theorem \ref{thm:main} is proportional to the \emph{aggregate} cache size  that is available throughout the network, even-though these caches are \emph{isolated}. 

\item (\textit{Equal contribution of transmitter and receiver caches}) Perhaps interestingly, the caches at both sides of the network, i.e., the transmitters' side and the receivers' side, are equally valuable in the achievable one-shot linear sum-DoF of the network. Note that in practice, size of each transmitter's cache,  $M_T$, could be large. However, the number of  transmitters (e.g., base stations) $K_T$  is often small. On the other hand, size of the cache $M_R$ at the receivers (e.g., cellphones) is small, whereas the number of receivers $K_R$ is large. Therefore $K_TM_T$ could be comparable with $K_RM_R$. Our result in Theorem \ref{thm:main} shows that neither caches at the transmitters nor caches at the receivers should be ignored.

\item (\textit{Linear scaling of DoF with network size}) Letting $K_T=K_R=K$, we observe that the one-shot linear sum-DoF scales \emph{linearly} with the number of users in a \emph{fully-connected} interference channel. Note that without caches, the one-shot linear sum-DoF of a fully-connected interference channel is bounded by 2, as shown in \cite{razaviyayn2012degrees}. Hence, caching enables linear growth of the DoF without the need for more complex physical layer schemes.

\item (\textit{Role of transmitter and receiver caches}) As we will show in Section \ref{sec:ach}, in~\eqref{eq:ach_dof}, $\frac{K_TM_T}{N}$ represents the contribution of collaborative zero-forcing at the transmitters' side, and $\frac{K_RM_R}{N}$ represents the gain of canceling the known interference at the receivers' side.

\item (\textit{Connection to single-server coded caching} \cite{maddahali_caching}) A special case of our network model is the case with a single transmitter, which was previously considered in \cite{maddahali_caching}. In this case, it can be shown that a sum-DoF of $\min\left\{1+\frac{K_R M_R}{N},K_R\right\}$ is achievable, which is equivalent to the global caching gain introduced in \cite{maddahali_caching}, indicating the number of receivers in the network that can be served simultaneously, interference-free. Hence, our result subsumes the result of \cite{maddahali_caching} by generalizing it for the case of multiple transmitters.

\item (\textit{Connection to multi-server coded caching} \cite{multi_server_caching}) Another special case of our network model is the case where each transmitter has space to cache the entire library; i.e., $M_T=N$. This case was previously considered in \cite{multi_server_caching} and it can be verified that in this case, a sum-DoF of $\min\left\{K_T+\frac{K_R M_R}{N},K_R\right\}$ is achievable. Hence, our result can also be viewed as a generalization of the result in \cite{multi_server_caching} where the cache size of each transmitter may be arbitrarily smaller than the entire library size.

\end{enumerate}

\begin{remk}
In practice, the files in the library have nonuniform demands and some of them are more popular than the rest. In this case, our algorithm can be used to cache and deliver the $N$ most popular files. If a user requests one of the remaining less popular files, it can be directly served by a central base station. The parameter $N$ can be tuned, based on the popularity pattern of the contents, in order to attain the best average performance.
\end{remk}

\begin{ex}
As an illustrative example, consider a cellular network with 5 base stations as transmitters, each with a 10 TB memory and 100 cellphones as receivers, each with a 32 GB memory. Moreover, consider a library of the 1000 most popular movie titles on Netflix, each with size of 5 GB. Then, Theorem \ref{thm:main} implies that at each time, around 11 cellphones can be served simultaneously interference-free, no matter what their demands are, in contrast to the naive time-sharing scheme, where at each time only 1 cellphone can be served.\hfill\qed
\end{ex}

The rest of the paper is devoted to the proof of Theorem \ref{thm:main}. In particular, we illustrate the achievable scheme in Section \ref{sec:ach} and we present the converse argument in Section \ref{sec:converse}.


\section{Achievable Scheme}\label{sec:ach}

In this section, we prove the achievability of Theorem \ref{thm:main} by presenting an achievable scheme which utilizes the caches at the transmitters and receivers efficiently to exploit the zero-forcing and interference cancellation opportunities at the transmitters' and receivers' sides, respectively. In particular, we introduce a prefetching strategy which maximizes the gains attained by the aforementioned opportunities in the delivery phase, no matter what the receiver demands are.

We first explain our achievable scheme through a simple, illustrative example and then proceed to mention our general achievable scheme.

\subsection{Description of the Achievable Scheme via an Example}\label{sec:scheme_example}

Consider a system with $K_T=3$ transmitters and $K_R=3$ receivers, where each transmitter has space to cache $M_T=2$ files and each receiver has space to cache $M_R=1$ file. The library has $N=3$ files $W_1=A$, $W_2=B$, and $W_3=C$, each consisting of $F$ packets. 
%

In the following, we will describe the prefetching and delivery phases in detail.

\textit{Prefetching Phase:} In this phase, each file $W_n, n\in[3]$ in the library is broken into $\binom{3}{2}\binom{3}{1}=9$ disjoint subfiles $W_{n,\mathcal{T},\mathcal{R}}$ for any $\mathcal{T}\subseteq[K_T]=[3]$ and $\mathcal{R}\subseteq[K_R]=[3]$ such that $|\mathcal{T}|=2$ and $|\mathcal{R}|=1$, where each subfile consists of $F/9$ packets. Each subfile $W_{n,\mathcal{T},\mathcal{R}}$ is then stored at the caches of the two transmitters in $\mathcal{T}$ and the single receiver in $\mathcal{R}$. For example, file $A$ is broken into 9 subfiles as follows:
\begin{align*}
A_{12,1},A_{12,2},A_{12,3},A_{13,1},A_{13,2},A_{13,3},A_{23,1},A_{23,2}, A_{23,3},
\end{align*}
where $A_{12,1}$ is stored at transmitters $\text{Tx}_1$ and $\text{Tx}_2$ as well as receiver $\text{Rx}_1$, $A_{12,2}$ is stored at transmitters $\text{Tx}_1$ and $\text{Tx}_2$ as well as receiver $\text{Rx}_2$, etc. We do the same partitioning for files $B$ and $C$, as well.

It is easy to verify that each transmitter caches 6 subfiles of each file, hence the total size of its cached content is $3*(6*F/9)=2F$ packets which satisfies its memory constraint. Also, each receiver caches 3 subfiles of each file and its total cached content has size $3*(3*F/9)=F$ packets, hence satisfying its memory constraint. Note that in this phase, we are unaware of receivers' future requests.

\textit{Delivery Phase:} In this phase, each receiver reveals its request for a file in the library. Without loss of generality, assume that receivers $\text{Rx}_1$, $\text{Rx}_2$ and $\text{Rx}_3$ request files $W_{d_1}=A$, $W_{d_2}=B$ and $W_{d_3}=C$, respectively. Note that each receiver has already stored 3 subfiles of its desired file in its own cache, and therefore the transmitters need to deliver the 6 remaining subfiles of each requested file. In particular, the following 18 subfiles need to be delivered by the transmitters to the requesting receivers:
\begin{align*}
A_{12,2},A_{12,3},A_{13,2},A_{13,3},A_{23,2},A_{23,3}&\text{ to receiver Rx}_1,\\
B_{23,3},B_{13,1},B_{12,3},B_{23,1},B_{13,3},B_{12,1}&\text{ to receiver Rx}_2,\numberthis\label{eq:remaining_subfiles}\\
C_{13,1},C_{23,2},C_{23,1},C_{12,2},C_{12,1},C_{13,2}&\text{ to receiver Rx}_3.
\end{align*}


We now show that we can break the 18 subfiles in \eqref{eq:remaining_subfiles} into 6 sets, each containing 3 subfiles, such that the subfiles in each set can be delivered simultaneously to the receivers, interference-free. Such a partitioning is illustrated through the 6 steps in Figure \ref{fig:ex1delivery}, where each step takes $\frac{F}{9}$ blocks. In each step, 3 subfiles are delivered to all the receivers simultaneously, while all the inter-user interference can be eliminated. For example, in the first step, as in Figure \ref{fig:ex1delivery}--(a), subfiles $A_{12,2},B_{23,3},C_{13,1}$ are respectively delivered to receivers $\text{Rx}_1$, $\text{Rx}_2$, and $\text{Rx}_3$ at the same time. In Figure \ref{fig:ex1ZF}, we show in detail how the interference is cancelled in this step. The transmit signals of transmitters $\text{Tx}_1$, $\text{Tx}_2$ and $\text{Tx}_3$ can be respectively written as
\begin{alignat*}{2}
X_1&=-&&h_{32} \tilde{A}_{12,2}+h_{23} \tilde{C}_{13,1},\\
X_2&=&&h_{31} \tilde{A}_{12,2} - h_{13} \tilde{B}_{23,3},\\
X_3&=-&&h_{21} \tilde{C}_{13,1} + h_{12} \tilde{B}_{23,3},
\end{alignat*}
where for any subfile $W_{n,\mathcal{T},\mathcal{R}}$, $\tilde{W}_{n,\mathcal{T},\mathcal{R}}$ denotes its coded version. For simplicity, in this example, we ignore the power constraint at the transmitters. On the other hand, the received signals by receivers $\text{Rx}_1$, $\text{Rx}_2$ and $\text{Rx}_3$ can be respectively written as
\begin{align*}
Y_1&=(h_{12} h_{31} -h_{11} h_{32}) \tilde{A}_{12,2}+(h_{11} h_{23}-h_{13} h_{21}) \tilde{C}_{13,1}+Z_1,\\
Y_2&=(h_{23} h_{12} -h_{22} h_{13}) \tilde{B}_{23,3}+(h_{22} h_{31}-h_{21} h_{32}) \tilde{A}_{12,2}+Z_2,\\
Y_3&=(h_{31} h_{23} -h_{33} h_{21}) \tilde{C}_{13,1}+(h_{33} h_{12}-h_{32} h_{13}) \tilde{B}_{23,3}+Z_3.
\end{align*}

Now, note that receivers $\text{Rx}_1$, $\text{Rx}_2$ and $\text{Rx}_3$ can cancel the interference due to $C_{13,1}$, $A_{12,2}$, and $B_{23,3}$, respectively, since they already have each respective subfile in their own cache. Therefore, all the interference in the network can be effectively eliminated and the receivers will be able to decode their desired subfiles.
Likewise, one can verify that all the receivers can receive their desired subfiles interference-free in all the 6 steps of communication depicted in Figure \ref{fig:ex1delivery}.

\begin{figure}[hbt]
\centering
\includegraphics[width=\textwidth]{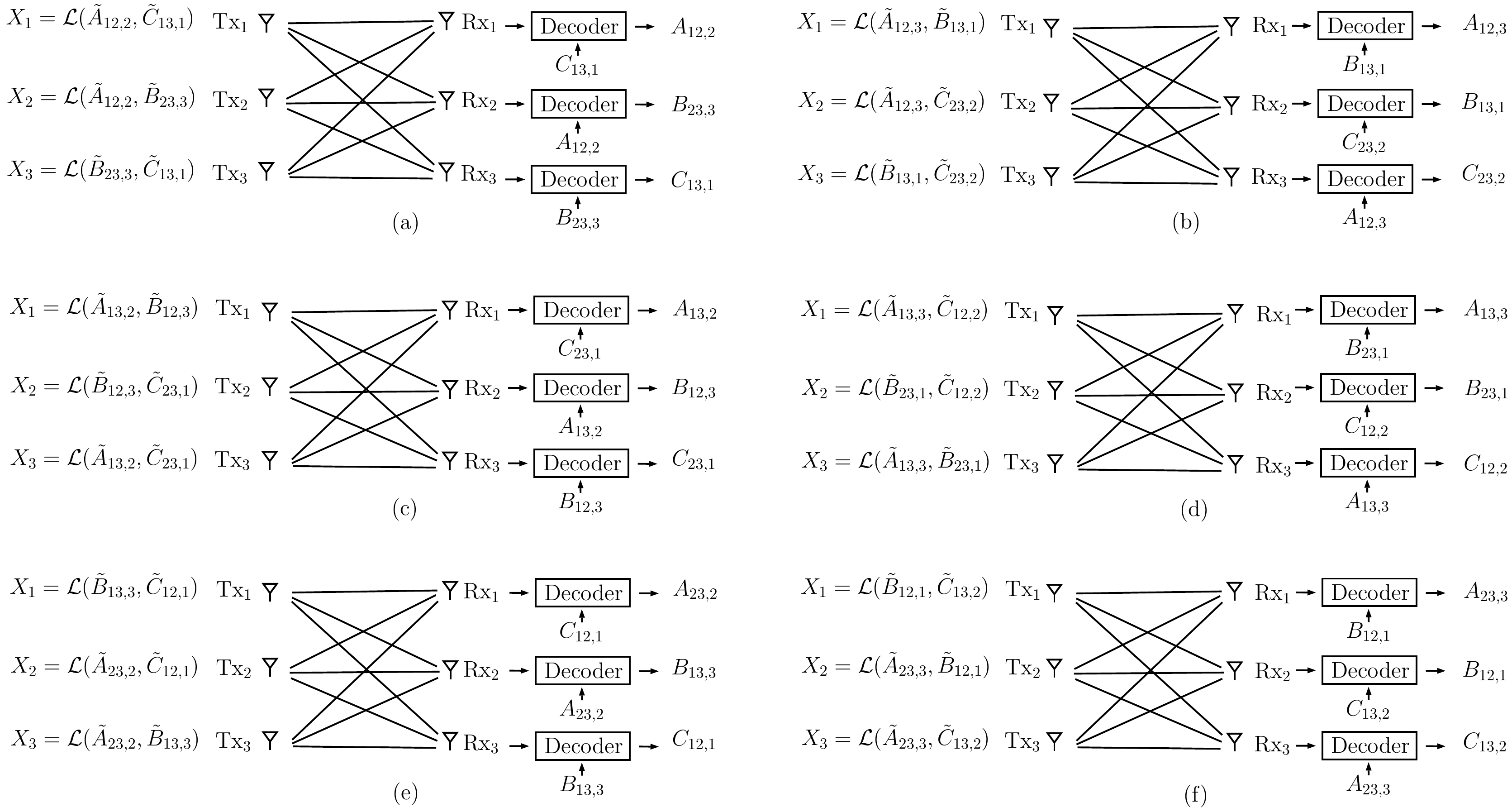}
\caption{Delivery phase for the example in Section \ref{sec:scheme_example} for respective requests of files $A$, $B$ and $C$ by receivers $\text{Rx}_1$, $\text{Rx}_2$ and $\text{Rx}_3$, where $\mathcal{L}(\alpha,\beta)$ denotes some linear combination of $\alpha$ and $\beta$. In every step, each pair of transmitters collaborate to zero-force the interference due to a specific subfile at a certain undesired receiver. Moreover, each receiver also uses its cache contents to cancel the interference due to the other interfering packet. Therefore, the communication is interference-free in all 6 steps.}
\label{fig:ex1delivery}
\end{figure}

\begin{figure}[hbt]
\centering
\includegraphics[width=.8\textwidth]{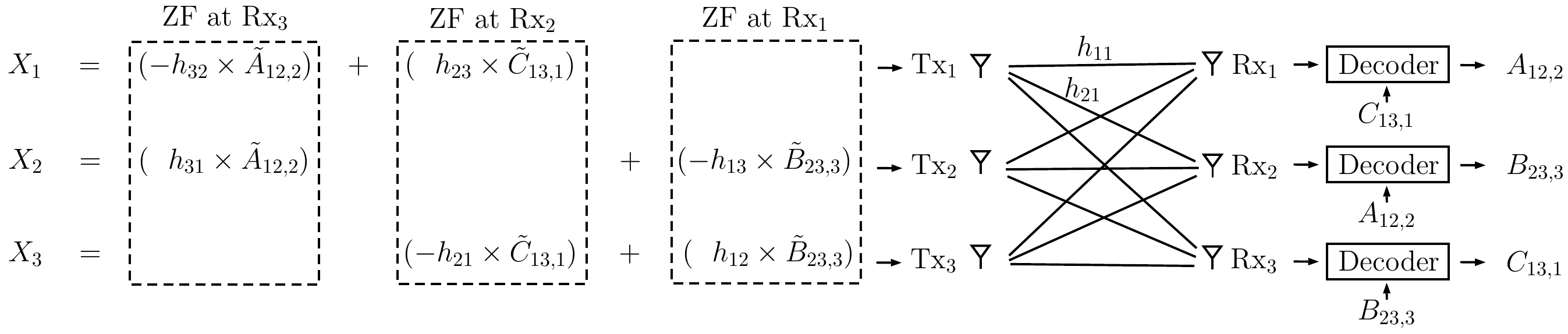}
\caption{More detailed description of the linear encoding and decoding schemes used in the delivery phase step in Figure \ref{fig:ex1delivery}--(a). In this step, $\text{Tx}_1$ and $\text{Tx}_2$ zero-force $A_{12,2}$ at $\text{Rx}_3$, $\text{Tx}_1$ and $\text{Tx}_3$ zero-force $C_{13,1}$ at $\text{Rx}_2$, and $\text{Tx}_2$ and $\text{Tx}_3$ zero-force $B_{23,3}$ at $\text{Rx}_1$. Moreover, $\text{Rx}_1$, $\text{Rx}_2$ and $\text{Rx}_3$ can cancel the interference due to $C_{13,1}$, $A_{12,2}$, and $B_{23,3}$, respectively, since they already have each respective subfile in their own cache.}
\label{fig:ex1ZF}
\end{figure}

Consequently, the 18 subfiles in \eqref{eq:remaining_subfiles}, each of which consists of $F/9$ packets, are delivered to the receivers in 6 steps, each consisting of $F/9$ blocks. Note that our particular file splitting pattern in the prefetching phase and the particular scheduling pattern in the delivery phase allows us to maximally exploit the two gains of zero-forcing the outgoing interference on the transmitters' side and canceling the known interference on the receivers' side, no matter what the receiver demands are in the delivery phase. Therefore, the sum-DoF of $\frac{18*F/9}{6*F/9}=3=\min\left\{\frac{K_T M_T+K_R M_R}{N},K_R\right\}$ is achievable in this network.


\subsection{Description of the General Achievable Scheme}\label{sec:ach_general}

Our general achievable scheme is given in Algorithm \ref{ach_alg}. In this algorithm, we use the notation
\begin{align}
t_T\triangleq  \frac{K_T M_T}{N},~~ t_R\triangleq \frac{K_R M_R}{N},
\end{align}
and for now, we assume that $t_T$ and $t_R$ are integers. Recall that in the example in Section \ref{sec:scheme_example}, $t_T=2$ and $t_R=1$.

In the following, we will describe the prefetching and delivery phases in more detail.
\begin{algorithm}
\caption{Achievable scheme for Theorem \ref{thm:main}}\label{ach_alg}
\begin{algorithmic}[1]
\Statex \textit{Prefetching Phase:}
\State \textbf{for} $n=1,...,N$
\State \quad Partition $W_n$ into $\binom{K_T}{t_T}\binom{K_R}{t_R}$ disjoint subfiles
${\left\{W_{n,\mathcal{T},\mathcal{R}}\right\}_{\mathcal{T}\subseteq[K_T], |\mathcal{T}|=t_T, \mathcal{R}\subseteq[K_R], |\mathcal{R}|=t_R}}$ of equal sizes.
\State \textbf{end}
\State \textbf{for} $i=1,...,K_T$
\State \quad \quad $\text{Tx}_i$ caches all $W_{n,\mathcal{T},\mathcal{R}}$ for which $i\in\mathcal{T}$.
\State \textbf{end}
\State \textbf{for} $j=1,...,K_R$
\State \quad \quad $\text{Rx}_j$ caches all $W_{n,\mathcal{T},\mathcal{R}}$ for which $j\in\mathcal{R}$.
\State \textbf{end}

\Statex

\Statex \textit{Delivery Phase:}
\State \textbf{for} $j \in [K_R]$
\State  \quad \textbf{for} ${\mathcal{T}\subseteq[K_T] \text{ s.t. } |\mathcal{T}|=t_T}$
\State \quad \quad \textbf{for} ${\mathcal{R}\subseteq[K_R] \setminus \{j\} \text{ s.t. } |\mathcal{R}|=t_R}$
\State \quad \quad \quad partition $W_{d_j,\mathcal{T},\mathcal{R}}$ to $\frac{t_R![K_R-(t_R+1)]!}{[K_R-(t_R+t_T)]!}$ disjoint subfiles $\Big\{ W_{d_j,\mathcal{T},\pi,\pi'} \Big\}_{\subalign{ &\pi\in \Pi_{\mathcal{R}} \\ &\pi' \in \Pi_{[K_R]\setminus (\mathcal{R}\cup \{j\}),t_T-1}}}$ of equal sizes.
\State \quad  \quad \textbf{end}
\State  \quad \textbf{end}
\State  \textbf{end}

\State \textbf{for} ${\mathcal{T}\subseteq[K_T] \text{ s.t. } |\mathcal{T}|=t_T}$
\State \quad \textbf{for} ${\mathcal{R}\subseteq[K_R] \text{ s.t. } |\mathcal{R}|=t_T+t_R}$
\State \quad \quad \textbf{for}
$\pi \in \Pi_{\mathcal{R}}^{\text{circ}}$
\State \quad \quad \quad Each transmitter $\text{Tx}_i$ transmits a linear combination of the coded subfiles as in
\Statex \quad \quad \quad $X_i=\mathcal{L}_{i,\mathcal{T},\pi}\left( \left\{\tilde{W}_{d_{\pi(l)},\mathcal{T}\oplus_{K_T}(l-1),\pi[l+1:l+t_R],\pi[l+t_R+1:l+t_R+t_T-1]}: l\in[t_T+t_R], i \in \mathcal{T} \oplus_{K_T} (l-1)\right\} \right)$
\Statex \quad \quad \quad using the linear combinations shown in Lemma \ref{lem:sch_packets} such that the subfiles
\Statex \quad \quad \quad $\left\{W_{d_{\pi(l)},\mathcal{T}\oplus_{K_T}(l-1),\pi[l+1:l+t_R],\pi[l+t_R+1:l+t_R+t_T-1]}: l\in[t_T+t_R]\right\}$ 
\Statex \quad \quad \quad  are simultaneously delivered to the receivers in $\mathcal{R}$ interference-free.
\State \quad \quad \textbf{end}
\State \quad \textbf{end}
\State \textbf{end}
\end{algorithmic}
\end{algorithm}



\subsubsection{Prefetching Phase}\label{sec:caching}
For any file $W_n$ in the library, $n\in[N]$, we partition it into $\binom{K_T}{t_T}\binom{K_R}{t_R}$ disjoint subfiles of equal sizes\footnote{Due to the assumption that $F$ is sufficiently large, we can assume that it is an integer multiple of $\binom{K_T}{t_T}\binom{K_R}{t_R}$.}, denoted by
\begin{align}\label{eq:subfile}
W_n=\Big\{W_{n,\mathcal{T},\mathcal{R}}\Big\}_{\substack{\mathcal{T}\subseteq[K_T]: |\mathcal{T}|=t_T\\ \mathcal{R}\subseteq[K_R]: |\mathcal{R}|=t_R}}.
\end{align}

Based on the above partitioning, in the prefetching phase, each transmitter $\text{Tx}_i$ stores a subset $\mathcal{P}_i$ of the packets in the library as described below.
\begin{align}\label{eq:txcache}
\mathcal{P}_i= \left\{W_{n,\mathcal{T},\mathcal{R}}:i\in\mathcal{T}\right\}.
\end{align}


\begin{illust}
For instance, in the example network considered in Section \ref{sec:scheme_example}, transmitter $\text{Tx}_3$ stores the following subset of packets in its cache.
\begin{align*}
\mathcal{P}_3=
\{&W_{1,13,1},W_{1,13,2},W_{1,13,3},W_{1,23,1},W_{1,23,2},W_{1,23,3},\\
&W_{2,13,1},W_{2,13,2},W_{2,13,3},W_{2,23,1},W_{2,23,2},W_{2,23,3},\\
&W_{3,13,1},W_{3,13,2},W_{3,13,3},W_{3,23,1},W_{3,23,2},W_{3,23,3}\}\\
=\{&A_{13,1},A_{13,2},A_{13,3},A_{23,1},A_{23,2},A_{23,3},\\
&B_{13,1},B_{13,2},B_{13,3},B_{23,1},B_{23,2},B_{23,3},\\
&C_{13,1},C_{13,2},C_{13,3},C_{23,1},C_{23,2},C_{23,3}\}.\tag*{\qed}
\end{align*}
\end{illust}

Based on the above caching strategy, we can verify that the total number of packets cached by transmitter $\text{Tx}_i$ equals
\begin{align*}
N\binom{K_T-1}{t_T-1}\binom{K_R}{t_R}\frac{F}{\binom{K_T}{t_T}\binom{K_R}{t_R}}=NF\frac{t_T}{K_T}=M_T F \text{ packets},
\end{align*}
hence satisfying its memory size constraint, where $\binom{K_T-1}{t_T-1}$ is the number of subsets $\mathcal{T}\subseteq[K_T]$ of size $t_T$ which include the transmitter index $i$.

Likewise, in the prefetching phase, each receiver $\text{Rx}_j$ stores a subset $\mathcal{Q}_j$ of the packets in the library as described below.
\begin{align}\label{eq:Q_j}
\mathcal{Q}_j= \left\{W_{n,\mathcal{T},\mathcal{R}}:j\in\mathcal{R}\right\}.
\end{align}

\begin{illust}
For instance, in the example network considered in Section \ref{sec:scheme_example}, receiver $\text{Rx}_2$ stores the following subset of packets in its cache.
\begin{align*}
\mathcal{Q}_2&=\{W_{1,12,2},W_{1,13,2},W_{1,23,2},
W_{2,12,2},W_{2,13,2},W_{2,23,2},
W_{3,12,2},W_{3,13,2},W_{3,23,2}\}\\
&=\{A_{12,2},A_{13,2},A_{23,2},
B_{12,2},B_{13,2},B_{23,2},
C_{12,2},C_{13,2},C_{23,2}\}.\tag*{\qed}
\end{align*}
\end{illust}

This suggests that the total number of packets cached by receiver $\text{Rx}_j$ is equal to
\begin{align*}
N\binom{K_T}{t_T}\binom{K_R-1}{t_R-1}\frac{F}{\binom{K_T}{t_T}\binom{K_R}{t_R}}=NF\frac{t_R}{K_R}=M_R F \text{ packets},
\end{align*}
which also satisfies its memory size constraint.

\subsubsection{Delivery Phase}\label{sec:delivery}
In this section, we first describe the delivery phase for the case where $t_T+t_R\leq K_R$, so that the first term in the lower bound in (\ref{eq:ach_dof}) is dominant. We will later show how to deal with the case where $t_T+t_R> K_R$.

In the delivery phase, the receiver requests are revealed, and in particular, each receiver $\text{Rx}_j, j\in[K_R]$ requests a file $W_{d_j}$ from the library and the transmitters need to deliver the subfiles in
\begin{align*}
\{W_{d_j,\mathcal{T},\mathcal{R}}: j\notin \mathcal{R}\}
\end{align*}
to receiver $\text{Rx}_j$; i.e., the subfiles of file $W_{d_j}$ which have not been already stored in the cache of receiver $\text{Rx}_j$.

In the following, our goal is to show that the set of packets which need to be delivered to the receivers can be partitioned into subsets of size $t_T+t_R$ such that the packets in each subset can be scheduled together. To this end, we need to further break each subfile to smaller subfiles. In particular, for any $j\in[K_R], \mathcal{T}\subseteq[K_T]$ s.t. $|\mathcal{T}|=t_T, \mathcal{R}\subseteq[K_R]\setminus\{j\}$ s.t. $|\mathcal{R}|=t_R$, we partition $W_{d_j,\mathcal{T},\mathcal{R}}$ to $\frac{t_R![K_R-(t_R+1)]!}{[K_R-(t_R+t_T)]!}$ smaller disjoint subfiles of equal sizes denoted by
\begin{align}\label{eq:subfile_further_partition}
W_{d_j,\mathcal{T},\mathcal{R}}=\Big\{ W_{d_j,\mathcal{T},\pi,\pi'} \Big\}_{\subalign{ &\pi\in \Pi_{\mathcal{R}} \\ &\pi' \in \Pi_{[K_R]\setminus (\mathcal{R}\cup \{j\}),t_T-1}}},
\end{align}
where for a set $\mathcal{S}$, $\Pi_{\mathcal{S}}$ denotes the set of permutations of $\mathcal{S}$, and for any $t\in\{1,...,|\mathcal{S}|\}$, $\Pi_{\mathcal{S},t}$ denotes the set of all permutations of all subsets of $\mathcal{S}$ of size $t$; i.e.,
\begin{align*}
\Pi_{\mathcal{S},t}=\bigcup_{\mathcal{A}\subseteq \mathcal{S}, |\mathcal{A}|=t} \Pi_{\mathcal{A}}.
\end{align*}

\begin{remk}
Note that in the example setting discussed in Section \ref{sec:scheme_example}, $\frac{t_R![K_R-(t_R+1)]!}{[K_R-(t_R+t_T)]!}=1$, which implies that further partitioning of the subfiles is not needed.
\end{remk}

The advantage of further breakdown of the subfiles in \eqref{eq:subfile_further_partition} is that we can now partition the set of the subfiles which need to be delivered to the receivers into certain subsets of size $t_T+t_R$ such that each subfile $W_{d_j,\mathcal{T},\pi,\pi'}$ intended for receiver $\text{Rx}_j$ is zero-forced at the receivers with indices in $\pi'$. Moreover, since this subfile is also already cached at the receivers with indices in $\pi$, the communication will be interference-free for each set of the $t_T+t_R$ subfiles.

We show how to do such a partitioning in Lemma \ref{lem:partition_packets}. In this lemma, we use the following notation: For a set $\mathcal{R}$, we let $\Pi_{\mathcal{R}}^{\text{circ}}$ denote the set of $(|\mathcal{R}|-1)!$ circular permutations of $\mathcal{R}$.\footnote{A circular permutation of a set $\mathcal{R}$ is a way of arranging the elements of $\mathcal{R}$ around a fixed circle. The number of distinct circular permutations of a set $\mathcal{R}$ is equal to $(|\mathcal{R}|-1)!$. For example, if $\mathcal{R}=\{1,2,3\}$, then $\Pi_{\mathcal{R}}^{\text{circ}}=\{[1,2,3],[1,3,2]\}$.} Moreover, for a set $\mathcal{S}$, a permutation $\pi \in \Pi_{\mathcal{S}}$ and two integers $i,j$ satisfying $j\geq i$, we define $\pi[i:j]$ as
\begin{align*}
\pi[i:j]=[\pi(i \: \oplus_{|\mathcal{S}|} 0) ~~ \pi(i \: \oplus_{|\mathcal{S}|} 1) ~~ \pi(i \: \oplus_{|\mathcal{S}|} 2) ~~ ... ~~ \pi(i \: \oplus_{|\mathcal{S}|} (j-i))],
\end{align*}
where for an integer $m$, $i \: \oplus_{m} j$ is defined as
\begin{align}\label{eq:mod_sum_def}
i \: \oplus_{m} j= 1+ (i+j-1 \mod m).
\end{align}

Finally, for a set $\mathcal{T}$ and an integer $j$, we let $\mathcal{T}\oplus_{m} j$ denote entry-wise addition of elements of $\mathcal{T}$ with $j$ modulo $m$, as defined in \eqref{eq:mod_sum_def}.

\begin{lem}\label{lem:partition_packets}
Given the prefetching phase in Section \ref{sec:caching}, for any receivers' demand vector $\mathbf{d}$, the set of subfiles which need to be delivered to the receivers can be partitioned into disjoint subsets of size $t_T+t_R$ as
\begin{align}\label{eq:subfile_perm}
\bigcup_{\substack{\mathcal{T}\subseteq[K_T]: |\mathcal{T}|=t_T \\ 
\mathcal{R}\subseteq[K_R]: |\mathcal{R}|=t_T+t_R \\
\pi \in \Pi_{\mathcal{R}}^{\emph{circ}}
}}
\left\{W_{d_{\pi(l)},\mathcal{T}\oplus_{K_T}(l-1),\pi[l+1:l+t_R],\pi[l+t_R+1:l+t_R+t_T-1]}: l\in[t_T+t_R]\right\}.
\end{align}
\end{lem}

\begin{proof}
See Appendix \ref{apx:proof_lem_partition}.
\end{proof}

\begin{illust}
For the example network mentioned in Section \ref{sec:scheme_example}, the set of 18 subfiles which need to be delivered to the receivers, as in \eqref{eq:remaining_subfiles}, can be partitioned to the following 6 sets.
\begin{align}
&\{A_{12,2},B_{23,3},C_{13,1}\} \cup
\{A_{12,3},B_{13,1},C_{23,2}\} \cup
\{A_{13,2},B_{12,3},C_{23,1}\} \nonumber \\ \cup \
&\{A_{13,3},B_{23,1},C_{12,2}\} \cup
\{A_{23,2},B_{13,3},C_{12,1}\} \cup
\{A_{23,3},B_{12,1},C_{13,2}\}.
\end{align}
\end{illust}

Based on the partitioning of the small subfiles that need to be delivered to the receivers in Lemma \ref{lem:partition_packets}, we will have $\binom{K_T}{t_T}\binom{K_R}{t_T+t_R} (t_T+t_R-1)!$ steps of communication, where at each step, specific sets $\mathcal{T}$ and $\mathcal{R}$ and a permutation $\pi$ are fixed as in \eqref{eq:subfile_perm}, and each transmitter $\text{Tx}_i$ will transmit a linear combination of the coded subfiles for which $i \in \mathcal{T} \oplus_{K_T} (l-1)$; i.e.,
\begin{align}\label{eq:tx_i_lin_comb}
X_i=\mathcal{L}_{i,\mathcal{T},\pi}\left( \left\{\tilde{W}_{d_{\pi(l)},\mathcal{T}\oplus_{K_T}(l-1),\pi[l+1:l+t_R],\pi[l+t_R+1:l+t_R+t_T-1]}: l\in[t_T+t_R], i \in \mathcal{T} \oplus_{K_T} (l-1)\right\} \right),
\end{align}
where for any subfile $W_{d_j,\mathcal{T},\pi,\pi'}$, $\tilde{W}_{d_j,\mathcal{T},\pi,\pi'}$ denotes the corresponding coded subfile containing PHY coded symbols, and $\mathcal{L}_{i,\mathcal{T},\pi}(.)$ represents the linear combination that transmitter $\text{Tx}_i$ chooses for sending the subfiles in \eqref{eq:tx_i_lin_comb}.

We will next show that under such a delivery scheme, there always exists a choice of linear combinations at the transmitters so that at each step, the communication will be interference-free and all the $t_T+t_R$ receivers in $\mathcal{R}$ can decode their desired packets, as we also showed in the example setting in Section \ref{sec:scheme_example}.



\begin{lem}\label{lem:sch_packets}
For any subset of $t_T$ transmitters $\mathcal{T}\subseteq [K_T]$, any subset of $t_T+t_R$ receivers $\mathcal{R}\subseteq[K_R]$, and any circular permutation $\pi \in \Pi_{\mathcal{R}}^{\emph{circ}}$, there exists a choice of the linear combinations $\{\mathcal{L}_{i,\mathcal{T},\pi}(.)\}_{i=1}^{K_T}$ in \eqref{eq:tx_i_lin_comb} such that the set of $t_T+t_R$ subfiles in
\begin{align}
\left\{W_{d_{\pi(l)},\mathcal{T}\oplus_{K_T}(l-1),\pi[l+1:l+t_R],\pi[l+t_R+1:l+t_R+t_T-1]}: l\in[t_T+t_R]\right\},
\end{align}
can be delivered simultaneously and interference-free by the transmitters in $\bigcup\limits_{l\in[t_T+t_R]} \big(\mathcal{T}\oplus_{K_T}(l-1)\big)$ to the receivers in $\mathcal{R}$.
\end{lem}

\begin{proof}
For ease of notation and without loss of generality, assume $$\mathcal{T}=\{1,...,t_T\},\mathcal{T}\oplus_{K_T}(l-1)=\{l,...,t_T+l\},\mathcal{R}=\{1,...,t_T+t_R\},\pi=[1,...,t_T+t_R].$$

First, we need to determine the subset of the subfiles which is available at each transmitter. It is easy to verify that


\begin{itemize}
\item If $i\in \{1,...,t_T-1\}$, then transmitter $\text{Tx}_i$ has subfiles
\begin{align}
\left\{\tilde{W}_{d_{\pi(l)},\mathcal{T}\oplus_{K_T}(l-1),\pi[l+1:l+t_R],\pi[l+t_R+1:l+t_R+t_T-1]}: l\in \{1,...,i\} \right\};
\end{align}

\item If $i\in \{t_T,...,t_T+t_R\}$, then transmitter $\text{Tx}_i$ has subfiles
\begin{align}
\left\{\tilde{W}_{d_{\pi(l)},\mathcal{T}\oplus_{K_T}(l-1),\pi[l+1:l+t_R],\pi[l+t_R+1:l+t_R+t_T-1]}: l\in \{i-t_T+1,...,i\} \right\};
\end{align}

\item and if $i\in \{t_T+t_R+1,...,2t_T+t_R-1\}$, then transmitter $\text{Tx}_i$ has subfiles
\begin{align}
\left\{\tilde{W}_{d_{\pi(l)},\mathcal{T}\oplus_{K_T}(l-1),\pi[l+1:l+t_R],\pi[l+t_R+1:l+t_R+t_T-1]}: l\in \{i-t_T+1,...,t_T+t_R\} \right\}.
\end{align}
\end{itemize}

Since each transmitter sends a linear combination of the subfiles that it has, the transmit signal of transmitter $\text{Tx}_i$ can be written as
\begin{align}
X_i=\begin{cases}
\ \ \: \: \sum\limits_{l=1}^{i} \ \ \: \: v_{i,l} \tilde{W}_{d_l,\{l,...,t_T+l\},\{l+1 ,...,l+t_R\},\{l+t_R+1 ,...,l+t_R+t_T-1\}},  & \text{ if } i\in \{1,...,t_T-1\}\\
\sum\limits_{l=i-t_T+1}^{i} v_{i,l} \tilde{W}_{d_l,\{l,...,t_T+l\},\{l+1 ,...,l+t_R\},\{l+t_R+1 ,...,l+t_R+t_T-1\}},  & \text{ if } i\in \{t_T,...,t_T+t_R\}\\
\sum\limits_{l=i-t_T+1}^{t_T+t_R} v_{i,l} \tilde{W}_{d_l,\{l,...,t_T+l\},\{l+1 ,...,l+t_R\},\{l+t_R+1 ,...,l+t_R+t_T-1\}},  & \text{ if } i\in \{t_T+t_R+1,...,2t_T+t_R-1\}
\end{cases}.
\end{align}

This implies that the received signal at receiver $\text{Rx}_j, j\in \{1,...,t_T+t_R\}$ can be written as
\begin{align}
Y_j&=\sum_{i=1}^{2t_T+t_R-1} h_{ji} X_i + Z_j\\
&=\sum_{i=j}^{t_T+j} h_{ji} v_{i,j} \tilde{W}_{d_j,\{j,...,t_T+j\},\{j+1 ,...,j+t_R\},\{j+t_R+1 ,...,j+t_R+t_T-1\}}\nonumber\\
&\quad+\sum_{l=j+1}^{j+t_T-1}\sum_{i=l}^{t_T+l} h_{ji} v_{i,l} \tilde{W}_{d_l,\{l,...,t_T+l\},\{l+1 ,...,l+t_R\},\{l+t_R+1 ,...,l+t_R+t_T-1\}}\nonumber\\
&\quad+\sum_{l=j-t_R}^{j-1}\sum_{i=l}^{t_T+l} h_{ji} v_{i,l} \tilde{W}_{d_l,\{l,...,t_T+l\},\{l+1 ,...,l+t_R\},\{l+t_R+1 ,...,l+t_R+t_T-1\}}+Z_j.\label{eq:rx_signal_ach}
\end{align}

Now, note that in \eqref{eq:rx_signal_ach}, the first term corresponds to the desired subfile of receiver $\text{Rx}_j$, while the second and third terms correspond to the undesired subfiles whose interference needs to be canceled at this receiver. However, note that the subfiles in the third term are already cached at receiver $\text{Rx}_j$ and hence it is able to cancel their incoming interference. Hence, in order for all receivers $\text{Rx}_j, j\in \{1,...,t_T+t_R\}$ to receive their subfiles interference-free, there should exist a choice of linear combination coefficients $\{v_{i,l}\}$ such that
\begin{align}
\sum_{i=j}^{t_T+j} h_{ji} v_{i,j} &= 1, \forall j\in \{1,...,t_T+t_R\}\label{eq:nulling_ach_1} \\
\sum_{i=l}^{t_T+l} h_{ji} v_{i,l} &= 0, \forall j\in \{1,...,t_T+t_R\}, \forall l\in \{j+1 ,...,j+t_T-1\}.\label{eq:nulling_ach_2}
\end{align}

Equations \eqref{eq:nulling_ach_1}-\eqref{eq:nulling_ach_2} introduce a system of $t_T(t_T+t_R)$ linear equations. On the other hand, the number of variables $\{v_{i,l}\}$ is also equal to $t_T(t_T+t_R)$. This indicates that there always exists a choice of linear combination coefficients $\{v_{i,l}\}$ such that \eqref{eq:nulling_ach_1}-\eqref{eq:nulling_ach_2} are satisfied. Finally, note that by scaling all the transmit signals by a large enough factor, the power constraint at all the transmitters can also be satisfied. Hence the proof is complete.
\end{proof}

\begin{remk}
As mentioned in Section \ref{sec:model}, we assume that the channel gains remain constant over the course of communication. However, for the delivery scheme presented in the proof of Lemma \ref{lem:sch_packets}, this assumption can be relaxed, since we only need the channel gains to remain unchanged for each block of communication and they can be allowed to vary among different blocks.
\end{remk}

\begin{remk}
In the delivery scheme presented in the proof of Lemma \ref{lem:sch_packets}, we only used zero-forcing at the transmitters in order to cancel their outgoing interference, which is DoF-optimal. However, in general one can use any scheme that exploits the collaboration among the transmitters in order to optimize the actual rates in the finite-SNR regime (such as the schemes suited for the MIMO broadcast channels~\cite{MIMOBCshamai}).
\end{remk}

\subsection{Analysis of the Sum-DoF of the Proposed Achievable Scheme}

As a result of Lemmas \ref{lem:partition_packets} and \ref{lem:sch_packets}, it is clear that for any set of receiver demands in the delivery phase, we can schedule all the requested subfiles in groups of size $t_T+t_R$. Now, if $t_T$ and/or $t_R$ are not integers, we can split the memories and the files proportionally so that for each new partition, the aforementioned scheme can be applied for updated $t_T$ and $t_R$ which are integers. Hence, combining the schemes over different partitions allows us to serve $t_T+t_R$ simultaneously, interference-free, for any values of $t_T$ and $t_R$ such that $t_T+t_R\leq K_R$.\footnote{In \cite{maddahali_caching}, this method is referred to as \emph{memory-sharing}, which resembles time-sharing in network information theory.} 

Finally, if $t_T+t_R>K_R$, then since we cannot serve more than $K_R$ receivers, we can neglect some of the caches at either the transmitters' side or the receivers' side and use a fraction of the caches with new sizes ${\frac{N}{K_T}\leq M'_T\leq M_T}$ and ${M'_R\leq M_R}$ so that ${\frac{K_T M'_T+K_R M'_R}{N}=K_R}$. We can then use Algorithm \ref{ach_alg} to serve all the $K_R$ receivers simultaneously without interference.

As we showed in Section \ref{sec:caching}, our prefetching phase respects the cache size constraint of all the transmitters and receivers. Moreover, given our prefetching phase, each receiver $\text{Rx}_j$ caches $\binom{K_T}{t_T}\binom{K_R-1}{t_R-1}\frac{F}{\binom{K_T}{t_T}\binom{K_R}{t_R}}=\frac{M_R}{N} F$ packets of each file in the library. Hence, for each set of requested files by the receivers, a total of $K_R \left(1-\frac{M_R}{N}\right) F$ packets need to be delivered by the transmitters to the receivers.


Therefore, based on the delivery phase mentioned in Section \ref{sec:ach_general}, the number of blocks required to deliver all the $K_R \left(1-\frac{M_R}{N}\right) F$ packets to the receivers is equal to $\frac{K_R \left(1-\frac{M_R}{N}\right) F}{\min\{t_T+t_R,K_R\}}$. This suggests that for any set of receiver demands, sum-DoF of $$\frac{K_R \left(1-\frac{M_R}{N}\right) F}{\frac{K_R \left(1-\frac{M_R}{N}\right) F}{\min\{t_T+t_R,K_R\}}}=\min\{t_T+t_R,K_R\}=\min\left\{\frac{K_T M_T+K_R M_R}{N},K_R\right\}$$ is achievable, hence completing the proof of achievability of Theorem \ref{thm:main}.

\section{Converse}\label{sec:converse}
In this section, we prove the converse of Theorem \ref{thm:main}. In particular, we show that the lower bound on the one-shot linear sum-DoF in (\ref{eq:ach_dof}) is within a factor of 2 of the optimal one-shot linear sum-DoF. In order to prove the converse, we take four steps as detailed in the following sections. First, we demonstrate how in each block of communication, the network can be converted into a virtual MISO interference channel. Second, we use this conversion to write an integer optimization problem for the minimum number of communication blocks needed to deliver a set of receiver demands for a given caching realization. Third, we show how we can focus on average demands instead of the worst-case demands to derive an outer optimization problem on the number of communication blocks optimized over the caching realizations. Finally, we present a lower bound on the value of the aforementioned outer optimization problem, which leads to the desired upper bound on the one-shot linear sum-DoF of the network.

\subsection{Conversion to a Virtual MISO Interference Channel}

Consider any caching realization ${\left(\{\mathcal{P}_i\}_{i=1}^{K_T},\{\mathcal{Q}_i\}_{i=1}^{K_R}\right)}$ and any demand vector $\mathbf{d}$. As discussed in Section \ref{sec:model}, in each communication blocks a subset of requested packets are selected to be sent to a corresponding subset of distinct receivers. Now, we can state the following lemma, which bounds the number of packets that can be scheduled together in a single communication block using a one-shot linear scheme.
\begin{lem}\label{lem:one_shot_ub}
Consider a single communication block where a set $\{\mathbf{w}_{n_l,f_l}\}_{l=1}^{L}$ of $L$ packets are scheduled to be transmitted together to $L$ distinct receivers. In order for each receiver to successfully decode its desired packet, the number of these concurrently-scheduled packets should be bounded by
\begin{align}\label{eq:one_shot_condition}
{L\leq \min_{l\in[L]} |\mathcal{T}_l|+|\mathcal{R}_{l}|,}
\end{align}
where for any $l\in[L]$, $\mathcal{T}_{l}$ and $\mathcal{R}_{l}$ denote the set of transmitters and receivers which have cached the packet $\mathbf{w}_{n_l,f_l}$, respectively.
\end{lem}

\begin{proof}
For ease of notation and without loss of generality, suppose that in the considered block, $L$ packets $\{\mathbf{w}_{1,1},...,\mathbf{w}_{L,1}\}$ are scheduled to be sent to $L$ receivers $\{\text{Rx}_1,...,\text{Rx}_L\}$, respectively. Each transmitter $\text{Tx}_i, i\in[K_T]$ will transmit
\begin{align}
\mathbf{x}_i=\sum_{{l: \: i\in\mathcal{T}_l}} v_{i,l,1} ~ \tilde{\mathbf{w}}_{l,1},
\end{align}
where we have dropped the dependency on the block index, since we are focusing on a single block.
On the other hand, the received signal of receiver $\text{Rx}_j, j\in[L]$ can be written as
\begin{align}
\mathbf{y}_j&=\sum_{i=1}^{K_T} h_{ji} \mathbf{x}_i\\
&=\sum_{i=1}^{K_T} h_{ji} \sum_{{l: \: i\in\mathcal{T}_l}} v_{i,l,1} ~ \tilde{\mathbf{w}}_{l,1}\\
&=\sum_{l=1}^L \sum_{i\in\mathcal{T}_l} h_{ji}~ v_{i,l,1} ~ \tilde{\mathbf{w}}_{l,1}.\label{eq:rx_signals_cache_rx}
\end{align}


Therefore, (\ref{eq:rx_signals_cache_rx}) implies that we can effectively convert the network into a new MISO interference channel with $L$ \emph{virtual} transmitters $\{\widehat{\text{Tx}}_l\}_{l=1}^L$, where $\widehat{\text{Tx}}_l$ is equipped with $|\mathcal{T}_{l}|$ antennas, and $L$ single-antenna receivers $\{\text{Rx}_j\}_{j=1}^{L}$, in which each virtual transmitter $\widehat{\text{Tx}}_l$ intends to send the coded packet $\tilde{\mathbf{w}}_{l,1}$ to receiver $\text{Rx}_l$. Each antenna in the new network corresponds to a transmitter in the original network. Hence, the channel vectors are correlated in the new network. In fact, as (\ref{eq:rx_signals_cache_rx}) suggests, all the antennas corresponding to the same transmitter in the original network have the same channel gain vectors to the receivers in the new network.




In the constructed MISO interference channel, we take a similar approach as in \cite{razaviyayn2012degrees} in order to bound the one-shot linear sum-DoF of the network. Each virtual transmitter $\widehat{\text{Tx}}_l$ in the constructed MISO network will select a beamforming vector $\mathbf{v}_l\in\mathbb{C}^{|\mathcal{T}_{l}| \times 1}$ (which consists of the coefficients chosen by the original transmitters corresponding to its antennas) to transmit its desired symbol. Denoting the channel gain vector between transmitter $\widehat{\text{Tx}}_l$ and receiver $\text{Rx}_j$ as $\mathbf{h}_{jl}\in\mathbb{C}^{|\mathcal{T}_{l}|\times 1}$, the decodability conditions can be written as
\begin{gather}
\mathbf{h}_{jl}^T \mathbf{v}_l=0, ~\forall l\neq j \text{ s.t. } j\notin \mathcal{R}_{l} \label{eq:algn_cache_rx}\\
\mathbf{h}_{jj}^T \mathbf{v}_j\neq 0, ~\forall  j\in [L].
\end{gather}

Now, each of the vectors $\mathbf{v}_l,~l\in[L]$ can be written as
\begin{align}
\mathbf{v}_l=q_l \mathbf{P}_l
\begin{bmatrix}
1 \\ \bar{\mathbf{v}}_l 
\end{bmatrix},
\end{align}
where $q_l$ is a non-zero scalar, $\mathbf{P}_l$ is a $|\mathcal{T}_{l}| \times |\mathcal{T}_{l}|$ permutation matrix and $\bar{\mathbf{v}}_l$ is a vector of size $(|\mathcal{T}_{l}|-1)\times 1$. Also, for any two distinct pairs $l \neq j$, the channel gain vector $\mathbf{h}_{jl}$ can be permuted as $\bar{\mathbf{h}}_{jl}=\mathbf{P}_l^{-1} \mathbf{h}_{jl}$, and we can partition $\bar{\mathbf{h}}_{jl}$ as
\begin{align}
\bar{\mathbf{h}}_{jl}=
\begin{bmatrix}
\bar{{h}}_{jl}^{(1)} \\ \bar{\mathbf{h}}_{jl}^{(2)}
\end{bmatrix},
\end{align}
where $\bar{{h}}_{jl}^{(1)}$ is a scalar and $\bar{\mathbf{h}}_{jl}^{(2)}$ is of size $(|\mathcal{T}_{l}|-1)\times 1$. Therefore, the nulling condition in (\ref{eq:algn_cache_rx}) can be rewritten as
\begin{align}\label{eq:algn2_cache_rx}
\begin{bmatrix}
\bar{{h}}_{jl}^{(1)} \bar{\mathbf{h}}_{jl}^{(2)}
\end{bmatrix}
\begin{bmatrix}
1 \\ \bar{\mathbf{v}}_{l} 
\end{bmatrix}=0 \Leftrightarrow \bar{{h}}_{jl}^{(1)}+\bar{\mathbf{h}}_{jl}^{(2)T}\bar{\mathbf{v}}_{l}=0.
\end{align}

Now, since the packet sent by the virtual transmitter $\widehat{\text{Tx}}_l$ is available in the caches of at most $|\mathcal{R}_{l}|$ receivers in the network, the interference of each transmitter should be nulled at least at $L-|\mathcal{R}_{l}|-1$ unintended receivers. This implies that the free beamforming variables at transmitter $l$, i.e., $\bar{\mathbf{v}}_l$, should satisfy at least $L-|\mathcal{R}_{l}|-1$ linear equations in the form of (\ref{eq:algn2_cache_rx}). This is not possible unless the number of equations is no greater than the number of variables, or
\begin{align}
L-|\mathcal{R}_{l}|-1 \leq |\mathcal{T}_{l}|-1 \Rightarrow L\leq |\mathcal{T}_{l}|+|\mathcal{R}_{l}|.
\end{align}

Since the above inequality holds for all $l\in[L]$, the proof is complete.
\end{proof}

\subsection{Integer Program Formulation}

Equipped with Lemma \ref{lem:one_shot_ub}, we define a set of packets $\mathcal{D}_m$ selected to be transmitted at block $m$ to be \emph{feasible} if its size satisfies condition \eqref{eq:one_shot_condition} in Lemma \ref{lem:one_shot_ub}. We can then write the following integer program (P1) to minimize the number of required communication blocks for any given caching realization and set of receiver demands:
\begin{align}
\min ~~~~&H \tag{P1-1}\\
\text{s.t.}~~~~& \bigcup_{m=1}^{H}\mathcal{D}_m = \bigcup_{j=1}^{K_R} \left( W_{d_j} \setminus \mathcal{Q}_j\right)\tag{P1-2}\label{eq:cover_all_packets}\\
&\mathcal{D}_m \text{ is feasible},\ \forall m \in [H],\tag{P1-3}
\end{align}
where (\ref{eq:cover_all_packets}) states that all the demanded packets that are not cached at the requesting receivers need to be delivered by the transmitters over the $H$ blocks of communication.

\subsection{Relaxing Worst-Case Demands to Average Demands and Optimizing over Caching Realizations}

We can now write an optimization problem to minimize the number of communication blocks required for delivering the worst-case demands optimized over the caching realizations. However, before that, we need to introduce some notation.

Given any caching realization ${\left(\{\mathcal{P}_i\}_{i=1}^{K_T},\{\mathcal{Q}_i\}_{i=1}^{K_R}\right)}$, we can break each file $W_n, n\in[N]$, in the library into $(2^{K_T}-1)(2^{K_R})$ subfiles $\{W_{n,\mathcal{T},\mathcal{R}}\}_{\mathcal{T}\subseteq_{\emptyset} [K_T],\mathcal{R}\subseteq [K_R]}$, where $W_{n,\mathcal{T},\mathcal{R}}$ denotes the subfile of $W_n$ exclusively stored in the caches of the transmitters in $\mathcal{T}$ and receivers in $\mathcal{R}$, and we use the shorthand notation $\mathcal{T}\subseteq_{\emptyset} [K_T]$ to denote $\mathcal{T}\subseteq [K_T],\mathcal{T}\neq \emptyset$. We define $a_{n,\mathcal{T},\mathcal{R}}$ as the number of packets in $W_{n,\mathcal{T},\mathcal{R}}$.




Denoting the answer to the optimization problem (P1) by ${H^*\left(\{\mathcal{P}_i\}_{i=1}^{K_T},\{\mathcal{Q}_i\}_{i=1}^{K_R}, \mathbf{d}\right)}$, the below optimization problem yields the number of communication blocks required for delivering the worst-case demands, minimized over all caching realizations:
\begin{align}
\min_{{\{\mathcal{P}_i\}_{i=1}^{K_T},\{\mathcal{Q}_i\}_{i=1}^{K_R}}} ~~~~&\max_{\mathbf{d}} ~~~~{H^*\left(\{\mathcal{P}_i\}_{i=1}^{K_T},\{\mathcal{Q}_i\}_{i=1}^{K_R}, \mathbf{d}\right)}\tag{P2-1}\label{eq:OPT_outer_start_rx_cache}\\
\text{s.t.}\qquad\qquad
&{\sum\limits_{\mathcal{T}\subseteq_{\emptyset}[K_T]} \sum\limits_{\mathcal{R}\subseteq[K_R]} a_{n,\mathcal{T},\mathcal{R}}=F,~\forall n\in [N]}\tag{P2-2}\\
&{\sum\limits_{n=1}^N \sum\limits_{\mathcal{R}\subseteq[K_R]} \sum\limits_{\substack{\mathcal{T}\subseteq [K_T]:\\i \in \mathcal{T}}} a_{n,\mathcal{T},\mathcal{R}}\leq M_T F, \forall i\in [K_T]}\tag{P2-3}\\
&\sum\limits_{n=1}^N \sum\limits_{\mathcal{T}\subseteq_{\emptyset}[K_T]} \sum\limits_{\substack{\mathcal{R}\subseteq[K_R]:\\j \in \mathcal{R}}} a_{n,\mathcal{T},\mathcal{R}}\leq M_R F, \forall j\in [K_R]\tag{P2-4}\\
&a_{n,\mathcal{T},\mathcal{R}}\geq 0, \forall n\in [N], \forall \mathcal{T}\subseteq_{\emptyset}[K_T],\forall \mathcal{R}\subseteq[K_R].\tag{P2-5}\label{eq:OPT_outer_end_rx_cache}
\end{align}

To lower bound the value of the above optimization problem, we can write the following optimization problem, which yields the number of communication blocks averaged over all the ${\pi(N,K_R)=\frac{N!}{(N-K_R)!}}$ permutations of distinct receiver demands, denoted by $\mathcal{P}_{N,K_R}$:
\begin{align}
\min_{{\{\mathcal{P}_i\}_{i=1}^{K_T},\{\mathcal{Q}_i\}_{i=1}^{K_R}}} ~~~~& \frac{1}{\pi(N,K_R)} {\sum_{\mathbf{d}\in\mathcal{P}_{N,K_R}} H^*\left(\{\mathcal{P}_i\}_{i=1}^{K_T},\{\mathcal{Q}_i\}_{i=1}^{K_R}, \mathbf{d}\right)}\tag{P3-1}\label{eq:OPT_outer_start_rx_cache_sum}\\
\text{s.t.}\qquad\qquad
&{\sum\limits_{\mathcal{T}\subseteq_{\emptyset}[K_T]} \sum\limits_{\mathcal{R}\subseteq[K_R]} a_{n,\mathcal{T},\mathcal{R}}=F,~\forall n\in [N]}\tag{P3-2}\label{eq:OPT_outer_start_rx_cache_sum_F_constraint}\\
& {\sum\limits_{n=1}^N \sum\limits_{\mathcal{R}\subseteq[K_R]} \sum\limits_{\substack{\mathcal{T}\subseteq [K_T]:\\i \in \mathcal{T}}} a_{n,\mathcal{T},\mathcal{R}}\leq M_T F, \forall i\in [K_T]}\tag{P3-3}\label{eq:OPT_cache_size_constraint_Tx}\\
& \sum\limits_{n=1}^N \sum\limits_{\mathcal{T}\subseteq_{\emptyset}[K_T]} \sum\limits_{\substack{\mathcal{R}\subseteq[K_R]:\\j \in \mathcal{R}}} a_{n,\mathcal{T},\mathcal{R}}\leq M_R F, \forall j\in [K_R]\tag{P3-4}\label{eq:OPT_cache_size_constraint_Rx}\\
& {a_{n,\mathcal{T},\mathcal{R}}\geq 0, \forall n\in [N], \forall \mathcal{T}\subseteq_{\emptyset}[K_T],\forall \mathcal{R}\subseteq[K_R].}\tag{P3-5}\label{eq:OPT_outer_end_rx_cache_sum}
\end{align}

\subsection{Lower Bound on the Number of Communication Blocks}

Having the optimization problem in (P3), we now present the following lemma which provides a lower bound on the value of (P3).

\begin{lem}\label{lem:OPT_bound_rx}
The value of the optimization problem (P3) is bounded from below by $\frac{K_R N F \left(1-\frac{M_R}{N}\right)^2}{K_T M_T+K_R M_R }$.
\end{lem}

\begin{proof}
See Appendix \ref{apx:proof_lem_OPT}.
\end{proof}

Since the total number of packets delivered over the channel is $K_R \left(1-\frac{M_R}{N}\right) F$ in the optimization problem (P3), Lemma \ref{lem:OPT_bound_rx} immediately yields the following upper bound on the one-shot linear sum-DoF:
\begin{align*}
{\DoF_{\text{L,sum}}^*(N,M_T,M_R)\leq \frac{K_R \left(1-\frac{M_R}{N}\right) F}{\frac{K_R N F \left(1-\frac{M_R}{N}\right)^2}{K_T M_T+K_R M_R }}=\frac{K_T M_T+K_R M_R}{N-M_R}.}
\end{align*}

Combining the above bound with the trivial bound on the one-shot linear sum-DoF which is the number of receivers, $K_R$, we have
\begin{align}\label{eq:final_bound}
{\DoF_{\text{L,sum}}^*(N,M_T,M_R)\leq \min\left\{\frac{K_T M_T+K_R M_R}{N-M_R},K_R\right\}.}
\end{align}


Now, consider the following two cases:
\begin{itemize}
\item $M_R\leq \frac{N}{2}$: In this case, (\ref{eq:final_bound}) implies that
\begin{align*}
{\DoF_{\text{L,sum}}^*(N,M_T,M_R)}&{\leq \min\left\{\frac{K_T M_T+K_R M_R}{N-\frac{N}{2}},K_R\right\}}\\
&{\leq \min\left\{2\frac{K_T M_T+K_R M_R}{N},K_R\right\}.}
\end{align*}

\item $M_R > \frac{N}{2}$: In this case, (\ref{eq:ach_dof}) implies that one-shot linear sum-DoF of
\begin{align*}
{\DoF_{\text{L,sum}} (N,M_T,M_R)> \min\left\{\frac{K_T M_T+K_R \frac{N}{2}}{N},K_R\right\}> \frac{K_R}{2},}
\end{align*}
can be achieved, while the upper bound in (\ref{eq:final_bound}) implies that ${\DoF_{\text{L,sum}}^*(N,M_T,M_R)\leq K_R}$.
\end{itemize}

Therefore, in both cases, the inner bound in (\ref{eq:ach_dof}) is within a factor of 2 of the outer bound in (\ref{eq:ach_dof}), which completes the proof of the converse of Theorem \ref{thm:main}.

\section{Concluding Remarks and Future Directions}\label{sec:conc}
In this work, we considered a wireless network setting with arbitrary numbers of transmitters and receivers, where all transmitters and receivers in the network are equipped with cache memories of specific sizes. We characterized the one-shot linear sum-DoF of the network to within a gap of 2. In particular, we showed that the one-shot linear sum-DoF of the network is proportional to the aggregate cache size in the network, even though the cache of each node is isolated from all the other nodes. We presented an achievable scheme which loads the caches carefully in order to maximize the opportunity for zero-forcing the outgoing interference from the transmitters and interference cancellation due to previously-cached content at the receivers. We also demonstrated that the achievable one-shot linear sum-DoF of our scheme is within a multiplicative factor of 2 of the optimal one-shot linear sum-DoF by bounding the number of communication blocks required to deliver any set of requested files to the receivers using an integer programming approach.

There are several interesting directions following this work. First, in this work we assumed all the links in the network to be present in the network topology. However, due to fading effects, some links between certain transmitter-receiver pairs might be absent from the network topology. It would be interesting to study what type of caching strategies are optimal in this case and to explore its connections to the index coding problem~\cite{birk,baryossef,rouayheb}. Another direction would be to combine caching with more sophisticated interference management schemes. Some initial results have been reported in \cite{cached_itlinq}, in which the authors used the replication in the cache contents at the transmitters in order to improve the system performance using the ITLinQ scheme \cite{itlinq,itlinq_dyspan,itlinq_isit}. It would be interesting to study the role of transmitter and receiver caches illustrated in this work in improving the achievable system throughput that more sophisticated delivery schemes such as ITLinQ can provide.

\appendices

\section{Proof of Lemma \ref{lem:partition_packets}}\label{apx:proof_lem_partition}

For any $\mathcal{T}\subseteq[K_T]$ s.t. $|\mathcal{T}|=t_T$ and for any $l\in[t_T+t_R]$, it is clear that the set $\mathcal{T}\oplus_{K_T}(l-1)$ is of size $t_T$. Also, for any $\mathcal{R}\subseteq[K_R]$ s.t. $|\mathcal{R}|=t_T+t_R$, and for any permutation $\pi \in \Pi_{\mathcal{R}}^{\text{circ}}$, the vector $\pi[l+1:l+t_R]$ is of size $t_R$ and the vector $\pi[l+t_R+1:l+t_R+t_T-1]$ is of size $t_T-1$.

Furthermore, note that $W_{d_{\pi(l)},\mathcal{T}\oplus_{K_T}(l-1),\pi[l+1:l+t_R],\pi[l+t_R+1:l+t_R+t_T-1]}$ is a subfile of the file $W_{d_{\pi(l)}}$ requested by receiver $\text{Rx}_{\pi(l)}$. However, since $\pi(l)\notin \pi[l+1:l+t_R]$, receiver $\text{Rx}_{\pi(l)}$ has not stored the packets in this subfile in its cache and therefore, this subfile needs to be delivered to this receiver.

Finally, each set inside the union in \eqref{eq:subfile_perm} is composed of $t_T+t_R$ subfiles. The number of such sets is equal to
\begin{align}
\binom{K_T}{t_T}\binom{K_R}{t_T+t_R} (t_T+t_R-1)!.
\end{align}
Hence, the total number of subfiles in \eqref{eq:subfile_perm} is equal to 
\begin{align}\label{eq:total_packets}
\binom{K_T}{t_T}\binom{K_R}{t_T+t_R} (t_T+t_R-1)!(t_T+t_R)=\binom{K_T}{t_T}\binom{K_R}{t_T+t_R} (t_T+t_R)!.
\end{align}

On the other hand, each receiver $\text{Rx}_j$ has already cached $\binom{K_T}{t_T}\binom{K_R-1}{t_R-1}$ subfiles as in \eqref{eq:Q_j} in its cache, and needs the rest of the subfiles of its requested file, i.e., $\binom{K_T}{t_T}\binom{K_R-1}{t_R}$ subfiles, where each subfile is further partitioned into $ \frac{t_R![K_R-(t_R+1)]!}{[K_R-(t_R+t_T)]!}$ smaller subfiles. Hence, the total number of small subfiles that need to be delivered to all the receivers is equal to
\begin{align}
K_R \left[ \binom{K_T}{t_T}\binom{K_R-1}{t_R} \right] \left[ \frac{t_R![K_R-(t_R+1)]!}{[K_R-(t_R+t_T)]!}\right]=\binom{K_T}{t_T}\binom{K_R}{t_T+t_R} (t_T+t_R)!,
\end{align}
which equals the total number of small subfiles in \eqref{eq:subfile_perm}, calculated in \eqref{eq:total_packets}.
Consequently, the set of requested subfiles which are not cached at the corresponding receivers can be partitioned as in \eqref{eq:subfile_perm}, hence the proof is complete. \hfill\qed

\section{Proof of Lemma \ref{lem:OPT_bound_rx}}\label{apx:proof_lem_OPT}

According to the constraint \eqref{eq:one_shot_condition}, each of the packets of \emph{order} $s$, which are available at $s$ nodes, either on the transmitter side or the receiver side, can be scheduled with at most $s-1$ packets of the same order. Therefore, for any given caching realization and set of demands, we have the lower bound
\begin{align*}
H^*\bigg(\{\mathcal{P}_i\}_{i=1}^{K_T},&\{\mathcal{Q}_i\}_{i=1}^{K_R}, \mathbf{d}\bigg)\\
&\geq \sum_{s=K_R}^{K_T+K_R} \sum_{j=1}^{K_R}  \sum_{\substack{\mathcal{T}\subseteq[K_T]:\\|\mathcal{T}|\in[s]}} \sum_{\substack{\mathcal{R}\subseteq[K_R]:\\|\mathcal{R}|=s-|S_T|\\j\notin\mathcal{R}}} \frac{a_{d_j,\mathcal{T},\mathcal{R}}}{K_R}+\sum_{s=1}^{K_R-1} \sum_{j=1}^{K_R} \sum_{\substack{\mathcal{T}\subseteq[K_T]:\\|\mathcal{T}|\in[s]}} \sum_{\substack{\mathcal{R}\subseteq[K_R]:\\|\mathcal{R}|=s-|S_T|\\j\notin\mathcal{R}}} \frac{a_{d_j,\mathcal{T},\mathcal{R}}}{s}\\
&\geq \sum_{s=1}^{K_T+K_R} \sum_{j=1}^{K_R} \sum_{\substack{\mathcal{T}\subseteq[K_T]:\\|\mathcal{T}|\in[s]}} \sum_{\substack{\mathcal{R}\subseteq[K_R]:\\|\mathcal{R}|=s-|S_T|\\j\notin\mathcal{R}}} \frac{a_{d_j,\mathcal{T},\mathcal{R}}}{s}.\numberthis
\end{align*}

Now, denoting the objective function in \eqref{eq:OPT_outer_start_rx_cache_sum} by $\bar{H}\left(\{\mathcal{P}_i\}_{i=1}^{K_T},\{\mathcal{Q}_i\}_{i=1}^{K_R}\right)$, we have
\begin{align*}
\bar{H}\bigg(\{\mathcal{P}_i\}_{i=1}^{K_T},\{\mathcal{Q}_i\}_{i=1}^{K_R}\bigg)
&\geq \frac{1}{\pi(N,K_R)} \sum_{s=1}^{K_T+K_R}  \frac{1}{s} \sum_{j=1}^{K_R} \sum_{\substack{\mathcal{T}\subseteq[K_T]:\\|\mathcal{T}|\in[s]}} \sum_{\substack{\mathcal{R}\subseteq[K_R]:\\|\mathcal{R}|=s-|S_T|\\j\notin\mathcal{R}}}  \pi(N-1,K_R-1) \sum_{n=1}^N a_{n,\mathcal{T},\mathcal{R}}\\
& = \frac{1}{N} ~ \sum_{s=1}^{K_T+K_R}  \frac{1}{s} \sum_{j=1}^{K_R} \sum_{\substack{\mathcal{T}\subseteq[K_T]:\\|\mathcal{T}|\in[s]}} \sum_{\substack{\mathcal{R}\subseteq[K_R]:\\|\mathcal{R}|=s-|S_T|\\j\notin\mathcal{R}}}   \sum_{n=1}^N a_{n,\mathcal{T},\mathcal{R}}\\
& = \frac{1}{N} ~ \sum_{r=1}^{K_T} \sum_{r'=0}^{K_R}  \frac{1}{r+r'} \sum_{j=1}^{K_R} \sum_{\substack{\mathcal{T}\subseteq[K_T]:\\|\mathcal{T}|=r}} \sum_{\substack{\mathcal{R}\subseteq[K_R]:\\|\mathcal{R}|=r'\\j\notin\mathcal{R}}}   \sum_{n=1}^N a_{n,\mathcal{T},\mathcal{R}}\\
& = \frac{1}{N} ~ \sum_{r=1}^{K_T} \sum_{r'=0}^{K_R}  \frac{K_R-r'}{r+r'} \sum_{\substack{\mathcal{T}\subseteq[K_T]:\\|\mathcal{T}|=r}} \sum_{\substack{\mathcal{R}\subseteq[K_R]:\\|\mathcal{R}|=r'}}   \sum_{n=1}^N a_{n,\mathcal{T},\mathcal{R}}\\
& = \frac{1}{N} ~ \sum_{r=1}^{K_T} \sum_{r'=0}^{K_R-1}  \frac{b_{r,r'}}{r+r'},\numberthis\label{eq:bound_avg_OPT}
\end{align*}
where for any $r\in[K_T]$ and $r'\in [K_R-1] \cup \{0\}$, we define
\begin{align}
b_{r,r'} \triangleq \sum_{j=1}^{K_R} \sum_{\substack{\mathcal{T}\subseteq[K_T]:\\|\mathcal{T}|=r}} \sum_{\substack{\mathcal{R}\subseteq[K_R]:\\|\mathcal{R}|=r'\\j\notin\mathcal{R}}}   \sum_{n=1}^N a_{n,\mathcal{T},\mathcal{R}} = (K_R-r') \sum_{\substack{\mathcal{T}\subseteq[K_T]:\\|\mathcal{T}|=r}} \sum_{\substack{\mathcal{R}\subseteq[K_R]:\\|\mathcal{R}|=r'}}   \sum_{n=1}^N a_{n,\mathcal{T},\mathcal{R}}.
\end{align}

Moreover, adding the constraint in (\ref{eq:OPT_cache_size_constraint_Tx}) over all transmitters yields
\begin{align}
K_T M_T F &\geq \sum_{i=1}^{K_T} \sum_{n=1}^N \sum_{\mathcal{R}\subseteq[K_R]} \sum_{\substack{\mathcal{T}\subseteq[K_T]:\\i \in \mathcal{T}}} a_{n,\mathcal{T},\mathcal{R}}\\
&=  \sum_{n=1}^N \sum_{\mathcal{R}\subseteq[K_R]} \sum_{i=1}^{K_T} \sum_{\substack{\mathcal{T}\subseteq[K_T]:\\i \in \mathcal{T}}} a_{n,\mathcal{T},\mathcal{R}}\\
&=  \sum_{n=1}^N \sum_{\mathcal{R}\subseteq[K_R]} \sum_{r=1}^{K_T} r \sum_{\substack{\mathcal{T}\subseteq[K_T]:\\|\mathcal{T}|=r}} a_{n,\mathcal{T},\mathcal{R}}.\label{eq:adding_tx_caches}
\end{align}

Likewise, adding the constraint in (\ref{eq:OPT_cache_size_constraint_Rx}) over all receivers yields
\begin{align}
K_R M_R F &\geq \sum_{j=1}^{K_R} \sum_{n=1}^N \sum_{\mathcal{T}\subseteq_{\emptyset}[K_T]} \sum_{\substack{\mathcal{R}\subseteq[K_R]:\\j \in \mathcal{R}}} a_{n,\mathcal{T},\mathcal{R}}\label{eq:adding_rx_caches1}\\
&=  \sum_{n=1}^N \sum_{\mathcal{T}\subseteq_{\emptyset}[K_T]} \sum_{j=1}^{K_R} \sum_{\substack{\mathcal{R}\subseteq[K_R]:\\j \in \mathcal{R}}} a_{n,\mathcal{T},\mathcal{R}}\\
&=  \sum_{n=1}^N \sum_{\mathcal{T}\subseteq_{\emptyset}[K_T]} \sum_{r'=0}^{K_R} r' \sum_{\substack{\mathcal{R}\subseteq[K_R]:\\|\mathcal{R}|=r'}} a_{n,\mathcal{T},\mathcal{R}},\label{eq:adding_rx_caches}
\end{align}
and from (\ref{eq:adding_tx_caches}) and (\ref{eq:adding_rx_caches}), we have
\begin{align}
(K_T M_T + K_R M_R) F&\geq \sum_{n=1}^N \left[ \sum_{\mathcal{R}\subseteq[K_R]} \sum_{r=1}^{K_T} r \sum_{\substack{\mathcal{T}\subseteq[K_T]:\\|\mathcal{T}|=r}} a_{n,\mathcal{T},\mathcal{R}} + \sum_{\mathcal{T}\subseteq[K_T]} \sum_{r'=0}^{K_R} r' \sum_{\substack{\mathcal{R}\subseteq[K_R]:\\|\mathcal{R}|=r'}} a_{n,\mathcal{T},\mathcal{R}} \right]\\
&=\sum_{r=1}^{K_T} \sum_{r'=0}^{K_R} (r+r') \sum_{\substack{\mathcal{T}\subseteq[K_T]:\\|\mathcal{T}|=r}} \sum_{\substack{\mathcal{R}\subseteq[K_R]:\\|\mathcal{R}|=r'}}   \sum_{n=1}^N a_{n,\mathcal{T},\mathcal{R}}\\
&\geq \sum_{r=1}^{K_T} \sum_{r'=0}^{K_R-1} \frac{r+r'}{K_R-r'} b_{r,r'}.\label{eq:ineq_sbs_rx_cache}
\end{align}

Now, using the Cauchy-Schwarz inequality, we can write
\begin{align}
\sum_{r'=0}^{K_R-1} b_{r,r'}
&\leq  \sqrt{\sum_{r'=0}^{K_R-1} \frac{r+r'}{K_R-r'} b_{r,r'}}
\sqrt{\sum_{r'=0}^{K_R-1} \frac{K_R-r'}{r+r'} b_{r,r'}}.
\end{align}

Summing the above inequality over $r$ yields
\begin{align}
\sum_{r=1}^{K_T} \sum_{r'=0}^{K_R-1} b_{r,r'}
&\leq \sum_{r=1}^{K_T} \left[ \sqrt{\sum_{r'=0}^{K_R-1} \frac{r+r'}{K_R-r'} b_{r,r'}}
\sqrt{\sum_{r'=0}^{K_R-1} \frac{K_R-r'}{r+r'} b_{r,r'}}\right]\\
&\leq   \sqrt{\sum_{r=1}^{K_T} \sum_{r'=0}^{K_R-1} \frac{r+r'}{K_R-r'} b_{r,r'}}
\sqrt{\sum_{r=1}^{K_T} \sum_{r'=0}^{K_R-1} \frac{K_R-r'}{r+r'} b_{r,r'}}\label{eq:2nd_cauchy}\\
&\leq \sqrt{(K_T M_T + K_R M_R) F}
\sqrt{\sum_{r=1}^{K_T} \sum_{r'=0}^{K_R-1} \frac{K_R-r'}{r+r'} b_{r,r'}},\label{eq:total_cache_plug_in}
\end{align}
where in \eqref{eq:2nd_cauchy} we have invoked the Cauchy-Schwarz inequality again and \eqref{eq:total_cache_plug_in} follows from (\ref{eq:ineq_sbs_rx_cache}). On the other hand, we have
\begin{align}
\sum_{r=1}^{K_T} \sum_{r'=0}^{K_R-1} b_{r,r'}
&=\sum_{r=1}^{K_T} \sum_{r'=0}^{K_R-1} \sum_{j=1}^{K_R} \sum_{\substack{\mathcal{T}\subseteq[K_T]:\\|\mathcal{T}|=r}} \sum_{\substack{\mathcal{R}\subseteq[K_R]:\\|\mathcal{R}|=r'\\j\notin\mathcal{R}}}   \sum_{n=1}^N a_{n,\mathcal{T},\mathcal{R}}\\
&=\sum_{r=1}^{K_T} \sum_{r'=0}^{K_R} \sum_{j=1}^{K_R} \left[\left(\sum_{\substack{\mathcal{T}\subseteq[K_T]:\\|\mathcal{T}|=r}} \sum_{\substack{\mathcal{R}\subseteq[K_R]:\\|\mathcal{R}|=r'}}   \sum_{n=1}^N a_{n,\mathcal{T},\mathcal{R}}\right)-\left(\sum_{\substack{\mathcal{T}\subseteq[K_T]:\\|\mathcal{T}|=r}} \sum_{\substack{\mathcal{R}\subseteq[K_R]:\\|\mathcal{R}|=r'\\j\in\mathcal{R}}}   \sum_{n=1}^N a_{n,\mathcal{T},\mathcal{R}}\right)\right]\\
&= K_R \left(\sum_{n=1}^N \sum_{\substack{\mathcal{T}\subseteq_{\emptyset}[K_T]}} \sum_{\substack{\mathcal{R}\subseteq[K_R]}}    a_{n,\mathcal{T},\mathcal{R}}\right)- \sum_{j=1}^{K_R} \sum_{n=1}^N \sum_{\substack{\mathcal{T}\subseteq_{\emptyset}[K_T]}} \sum_{\substack{\mathcal{R}\subseteq[K_R]:\\j\in\mathcal{R}}}    a_{n,\mathcal{T},\mathcal{R}}\\
&\geq K_R(N-M_R)F,\label{eq:whole_uncached_lib}
\end{align}
where the inequality is due to \eqref{eq:OPT_outer_start_rx_cache_sum_F_constraint} and \eqref{eq:adding_rx_caches1}. Therefore, we can continue \eqref{eq:bound_avg_OPT} to bound the objective function in \eqref{eq:OPT_outer_start_rx_cache_sum}  as
\begin{align}
\bar{H}\left(\{\mathcal{P}_i\}_{i=1}^{K_T},\{\mathcal{Q}_i\}_{i=1}^{K_R}\right)   &\geq \frac{1}{N} ~ \sum_{r=1}^{K_T} \sum_{r'=0}^{K_R-1}  \frac{b_{r,r'}}{r+r'}\\
&\geq \frac{1}{K_R N} ~ \sum_{r=1}^{K_T} \sum_{r'=0}^{K_R-1}  \frac{(K_R-r')b_{r,r'}}{r+r'}\\
& \geq \frac{1}{ K_R N F(K_T M_T+K_R M_R) } \left(\sum_{r=1}^{K_T} \sum_{r'=0}^{K_R-1} b_{r,r'}\right)^2\label{eq:final_step_cauchy}\\
&\geq \frac{1}{ K_R N F(K_T M_T+K_R M_R)} \bigg(K_R(N-M_R)F\bigg)^2\label{eq:final_step_whole_uncached_lib}\\
&=\frac{K_R N F \left(1-\frac{M_R}{N}\right)^2}{K_T M_T+K_R M_R },
\end{align}
where \eqref{eq:final_step_cauchy} and \eqref{eq:final_step_whole_uncached_lib} follow from \eqref{eq:total_cache_plug_in} and \eqref{eq:whole_uncached_lib}, respectively. This completes the proof. \hfill\qed

\bibliographystyle{IEEEtran}
{\footnotesize
\bibliography{navid}}

\end{document}